\documentclass[10pt,leqno]{article}
\usepackage{graphicx,psfrag,amssymb,amsmath,amsthm,bbm, MnSymbol}
\usepackage[dvips]{color}
\usepackage{psfrag}
\usepackage{eucal}

\newtheorem{theorem}{Theorem}

\newtheorem{proposition}{Proposition}
\newtheorem{corollary}{Corollary}
\newtheorem{definition}{Definition}

\numberwithin{equation}{section}

\begin{document}

\title{\bf Random billiards with wall temperature and associated Markov chains}
\author{Scott Cook\footnote{Department of Mathematics \& Statistics, Swarthmore College, Swarthmore, PA 19081},
Renato Feres\footnote{Department of Mathematics, Washington University, Campus Box 1146, St. Louis, MO 63130}}
\date{\today}

\maketitle

\begin{abstract}
By a \emph{random billiard} we mean a billiard system  in which the standard specular reflection rule 
is replaced with a Markov transition probabilities operator  $P$ that, at each collision of the
billiard particle with the boundary of the billiard domain, gives the probability distribution of 
the post-collision velocity for a  given pre-collision velocity. A  random billiard with \emph{microstructure}, or RBM for short,  is a random billiard for which $P$  is derived from  a choice of  geometric/mechanical structure on the boundary 
of the billiard domain. 
Such random billiards     provide simple and explicit mechanical models   of
particle-surface interaction    that can incorporate thermal effects  and permit a detailed study of thermostatic action
from the perspective  of the standard theory of Markov chains on general state spaces. 

The main  focus of the present paper    is on  the operator 
 $P$ itself  and how it relates to the mechanical and geometric features of the microstructure,   such as
  mass ratios, curvatures, and potentials. The main results are as follows:
 (1) we characterize the stationary probabilities  (equilibrium states) of $P$ and show how standard equilibrium distributions studied in classical statistical mechanics, such as
 the \emph{Maxwell-Boltzmann distribution} and the 
 \emph{Knudsen cosine law}, arise naturally 
as generalized invariant billiard measures;
 (2) we obtain some basic functional theoretic properties of $P$.  Under very general conditions, we show that $P$ is a self-adjoint operator of norm $1$ on an appropriate Hilbert space.  In a simple but illustrative example, we show that $P$ is a compact (Hilbert-Schmidt) operator. This leads to the issue of relating the spectrum of eigenvalues of $P$ to the geometric/mechanical features of the billiard microstructure;
 (3) we explore the latter issue both analytically and numerically in a few representative examples;
 (4) we present a general algorithm for simulating these Markov chains based on a geometric description of the invariant volumes of classical statistical mechanics. Our description of these volumes may also have independent interest.
 \end{abstract}

\section{Introduction}
This extended introduction contains  some     definitions and an overview of the 
main results. Additional results and refinements are discussed throughout the text. 

\subsection{Physical motivation}
Consider the idealized  and somewhat fanciful  billiard system shown in Figure \ref{fancybilliard}. 
At a ``macroscopic scale'' it consists of a point particle, henceforth called  the {\em molecule},  and a billiard table having
piecewise smooth wall (only  a small part of which is shown). 
At a ``microscopic scale,'' both the wall and the molecule may reveal further geometric and mechanical structure
that can affect the outcome of a collision.
Thus collisions are not necessarily specular; to specify the  outcome of a collision  it is necessary to
consider  the interaction  between  molecule and wall at this finer scale.
 We suppose that
the wall system is kept at  a constant statistical state, say, a canonical
ensemble distribution with a given temperature, and wish to follow the evolution of the statistical state of
the molecule. The outcome of a molecule-wall {\em collision event} is then shown to be
 described by a time-independent transition probabilities operator $P$, to be  defined later as an operator on an $L^2$ space 
 over  the set  of  pure states of the molecule. This operator, which as we will see  is canonically defined by 
 the mechanical/geometric features of the wall and molecule microstructure and the constant statistical state of the wall, 
   replaces the mirror-reflection map  of
an ordinary billiard. We call  a system of this kind a {\em random billiard with microstructure}, or RBM.

\begin{figure}[htbp]
\begin{center}
\includegraphics[width=4in]{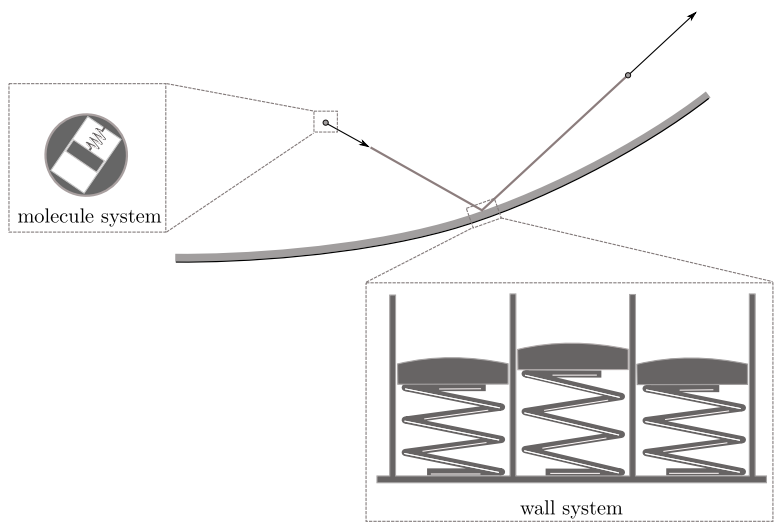}\ \ 
\caption{\small An arbitrary  molecule-wall system defining  a random billiard with microstructure. }
\label{fancybilliard}
\end{center}
\end{figure}

  Discrete-time Markov chains associated to $P$ are interpreted as  random  states of the molecule 
immediately after each collision, starting from some initial probability distribution,  at least for simple shapes
of the
the billiard  domain such as  cylinders or balls.  
Besides determining  the  equilibrium states of
the molecule as the
stationary (i.e., $P$-invariant) probability distributions for  the Markov chain,   the   operator $P$ contains 
 information about rates of decay of correlations  and spectral data, which can in principle
 be used to derive  transport coefficients such as the diffusion constant of a
 gas of non-interacting molecules moving inside a  ``billiard channel.'' (See Figure \ref{channel}.)

Here we   develop some of these ideas in detail,  focusing    on  the billiard-Markov
 operator $P$  and its relationship with the microstructure.  The
main results of the paper are concerned with defining   $P$ for any  given Newtonian molecule-wall system, 
deriving its basic   functional analytic properties,   describing stationary probability distributions,
and illustrating with concrete examples some of the spectral properties of P.

There are several sources of motivations for this work, some purely mathematical and others more applied.
On the purely mathematical side, we seek to have interesting and well-motivated classes of
  Markov chains that can be used to investigate issues of general interest in probability theory, such as
spectral gap, mixing times, central limit theorems, etc.  The   statistical mechanics perspective
   in combination with very simple mechanical systems  provides a great variety of   examples.
We also believe that generalized billiard systems of the kind we are considering may provide 
fruitful examples of random  (often hyperbolic) dynamical systems with singularities, i.e., random counterparts of  
the widely studied
chaotic billiards, for which \cite{CM} is a recommended reference. 
On the more applied side,   processes  of the kind we are studying here may be useful in kinetic theory of gases as suggested, for example,
in \cite{FY} in the context of Knudsen diffusion studies.   In the context of 
the theory of Boltzmann equations, operators such as our  $P$  may serve to specify  boundary conditions for
gas-wall systems. (See, for example,  \cite{Cer} for the context in the theory of  Boltzmann equations  in which related operators, but
not derived from any explicit microscopic interaction model  like ours,
arise.) Our random collision operators provide very natural and simple Newtonian  models
for  the interaction of a molecule with a heat bath that can be used to study thermostatic action  fairly explicitly and
often analytically  from the perspective of the general  theory of Markov chains.

\subsection{The   surface-scattering set-up}\label{setupintro}
The main definitions pertaining to the billiard microstructure are as follows. 
Let $M$ be a smooth manifold whose points represent the   configurations of a mechanical   system.
The system consists of 
two interacting subsystems: the {\em wall}  and the {\em molecule}.
A motion in $M$  describes  a    molecule-wall  {\em collision event} in which the molecule comes 
close to the  the wall surface, scatters off of it, and moves away.
Let smooth Riemannian manifolds $M_{\text{\tiny wall}}$ and  $M_{\text{\tiny mol}}$ be    the configuration spaces of
the {\em wall} and the {\em molecule} subsystems; to capture the idea that there is a direction towards which the (center of mass of)  the
molecule approaches the  plane of the wall, and that the microstructure on the wall is periodic, we assume that
 $M_{\text{\tiny mol}}$ factors as a Riemannian product
$$M_{\text{\tiny mol}}= \overline{M}_{\text{\tiny mol}}\times \mathbbm{R}\times \mathbbm{T}^k,$$
where  $\overline{M}_{\text{\tiny mol}}$ is the manifold of molecular configurations under the assumption 
that the center of mass is at a fixed position. For the examples of interest, 
 $k\leq 2$.

In the example of   Figure \ref{fancybilliard}, 
 $M_{\text{\tiny wall}}$ is simply an interval $[0,l]$, representing the range of positions
of the wall-bound mass attached to the spring. The Riemannian metric 
on $M_{\text{\tiny wall}}$ is   derived from the kinetic energy of the wall-bound mass.
The manifold $ \overline{M}_{\text{\tiny mol}}$ may be written as  $SO(2)\times [0, h]$, specifying 
the spatial orientation (or angle of rotation) of the hollowed little disc and the position of the vibrating  mass in its interior.  The Riemannian
metric is, again, derived from the kinetic energy of the molecule system, so metric coefficients  are given by 
the values of  masses and moments of inertia. The plane of the wall is aligned with the factor $\mathbbm{T}^1$ and
the direction of approach of the molecule is the factor $\mathbbm{R}$ in the product.
We disregard the possible (``macroscopic'') curvature of the billiard table boundary\----the interaction
 is imagined to happen at a length scale in which   boundary curvature cannot be discerned.

Back to the general case,
the combined {\em wall-molecule system} is represented by  $M$ with 
a Riemannian metric  
  and   a {\em potential function } $U:M\rightarrow \mathbbm{R}$ such that   (1) the two subsystems
are   non-interacting  when they are sufficiently far apart   (more details below)
and, (2) for each value $\mathcal{E}$  of the total energy $E$, where $E(q,v)=\frac12\|v\|^2 +U(q)$ for  $(q,v)\in TM$,
 the subset of the level set  $E^{-1}(\mathcal{E})$  consisting of states
at  which  the subsystems are a bounded distance  from each other has finite volume with respect
to the invariant   volume form $\Omega^E$, whose definition  is recalled later.

To explain assumption (1), we assume the existence of  a  smooth function
 $d:M\rightarrow \mathbbm{R}$, interpreted as the distance in Euclidian space from the center of mass of the molecule
  to some reference position on  the wall.
For each real number $a$ define 
 $M_{\text{\tiny mol}}(a):=\overline{M}_{\text{\tiny mol}}\times (a,\infty)\times \mathbbm{T}^k$ and
 $M(a)=M_{\text{\tiny mol}}(a)\times M_{\text{\tiny wall}}$, and 
  denote   by    $\pi_{\text{\tiny wall}}$ and $ \overline{\pi}_{{\text{\tiny mol}}} $  the projections onto
  $ M_{\text{\tiny wall}}$ and  $\overline{M}_{\text{\tiny mol}}$, respectively. 
 Then we suppose that there exists an $a_0\in \mathbbm{R}$, 
 which can be taken with no loss of generality to be  less than $0$,
 such that, for all $a\geq a_0$,
 \begin{enumerate}
 \item[i.] the set   $\{q\in M:d(q)>a\}$
 is isometric to, and will be identified with,  $M(a);$ 
 \item[ii.] there are smooth functions  $U_{\text{\tiny mol}}:\overline{M}_{\text{\tiny mol}}\rightarrow \mathbbm{R}$ and
 $U_{\text{\tiny wall}}:{M}_{\text{\tiny wall}}\rightarrow \mathbbm{R}$  
 such that
 $$U|_{M({a})} =U_{\text{\tiny mol}}\circ \overline{\pi}_{{\text{\tiny mol}}}  +  U_{\text{\tiny wall}}\circ\pi_{\text{\tiny wall}} $$
 \item[iii.] for each value $\mathcal{E}$ of the energy function  
 $ E,$
 the    level set 
$\{v\in T(M\setminus M(a)):E(v)=\mathcal{E}\}$ has finite volume relative to
  $\Omega^E$;
 \item[iv.] the system is {\em essentially dynamically complete}, in the following sense: 
 Any smooth curve $t\mapsto c(t)$
 that satisfies Newton's equation (with acceleration defined in terms of the Levi-Civita connection)
 $$ \frac{\nabla c'(t)}{dt}=-\text{grad}_{c(t)} U$$
 can be extended indefinitely in the interior of $M$,   until it reaches  the boundary;
 whenever $c$ intersects  the boundary transversely at a regular point $q=c(t)$, it can be extended further back into the interior
 along the unique solution curve with  initial state $(q,w)$, where $w=R_q c'(t)$ and $R_q:TM\rightarrow TM$ is the standard
 reflection map.
 \end{enumerate}

In the example of
 Figure \ref{exampleV},  $\overline{M}_{\text{\tiny mol}}$ is the two-point set $\{-1,1\}$ labeling the two sheets of $M$ above  a certain distance from the
handles. The manifold $M_{\text{\tiny wall}}$ consists of a single point and the potential function $U$ is constant.  We give $M$,
say, the   Riemannian metric induced from Euclidean $3$-space. 
More representative examples, in which $M_{\text{\tiny wall}}$ is non-trivial will be shown later.

\begin{figure}[htbp]
\begin{center}
\includegraphics[width=3.5in]{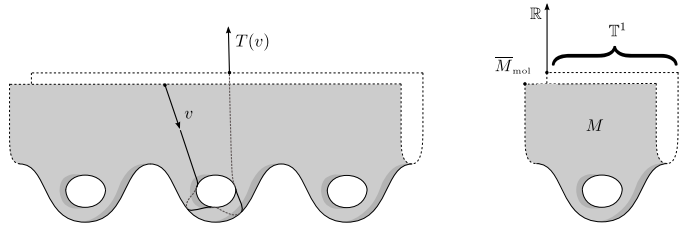}\ \ 
\caption{\small Geodesic motion on a periodic surface representing a molecule-wall scattering process.}
\label{exampleV}
\end{center}
\end{figure}

As already noted, the various   manifolds above may have boundary.
Boundary points   represent collision configurations. 
It is necessary  to accept  manifolds whose boundaries may  not be  smooth. For concreteness, we adopt here the class of {\em manifolds
with corners} (see \cite{lee}), which is general enough to provide plenty of  meaningful examples.
In particular, $M$   contains  a  set $\partial_{\text{\tiny s}} M$, the {\em singular boundary}, the  complement of which is a smooth manifold with
  boundary in the ordinary sense of being modeled on open subsets of the upper half space.
  This complement is  the union of the  interior  set 
 $M^\circ$, and the   (regular) boundary  $\partial_{\text{\tiny r}}M$. Moreover,  $\partial_{\text{\tiny s}}M$
 is contained in the closure of $\partial_{\text{\tiny r}}M$, it is  nowhere dense in this  closure
 and has measure $0$ in $\partial M$.  Since we are mainly interested in probabilistic questions, 
 it is typically  safe to  ignore the singular boundary set.

If $q$ is a regular boundary point of $M$ and $\nu$ is a unit vector perpendicular to the boundary at $q$,
we   assume that a motion in $M$ is extended after hitting the boundary at $q$ in such a way
that the pre- and post-collision velocities $v$ and $v'$ are related according to
the standard linear (reflection) map 
$v\mapsto v':= v -2 \langle v, \nu\rangle_q \nu$;
so    ``microscopic'' collisions are specular. Being an isometry of the kinetic energy metric, this map  leaves the energy 
function $E$ invariant.

Let $S$ be the level set $d=0$ in $M$, i.e.,
$$S=\overline{M}_{\text{\tiny mol}}\times \{0\}\times \mathbbm{T}^k\times M_{\text{\tiny wall}}$$
  and   $N_S$ the restriction  of $TM$ to $S$ (more precisely, the pull-back of $TM$ under the inclusion $S\hookrightarrow M$).
Informally, crossing $S$ amounts to entering the zone of interaction $M\setminus M(a_0)$, though S itself lies in the product zone. The vectors in $N_S$ pointing into the
zone of interaction form the subset $N^+_S$. This is the
set of {\em incoming states}. The set of {\em outgoing states}, $N^-_S$, is similarly defined as the set of vectors  in $N_S$ pointing
out of the zone of interaction.
Omitting, as we often do,  the base point in $M$ when referring to a state in $N_S$,  then $v\mapsto -v$
sends an element of $N^-_S$ to an element of $N^+_S$.
Let
$\mathbbm{H}\times\mathbbm{R}\times\mathbbm{T}^k=T(\mathbbm{R}\times\mathbbm{T}^k)\cap N^+_S$,
where $\mathbbm{H}$ is the half-space in  $\mathbbm{R}^{k+1}$. Then the  incoming states decompose  as
a product 
$$N^+_S=N_{\text{\tiny mol}}\times \mathbbm{T}^k\times N_{\text{\tiny wall}}$$ 
where $ N_{\text{\tiny wall}}:= TM_{\text{\tiny wall}}$ and 
$N_{\text{\tiny mol}}:=T\overline{M}_{\text{\tiny mol}}\times \mathbbm{H}$.  
We have chosen this particular decomposition so that the ``observable'' quantities of the molecule are grouped into the first factor and the quantities to be chosen probabilistically are grouped into the second and third factors.

A {\em collision event} is defined by an application of  the map $T:N^+_S\rightarrow N^-_S$, which
gives the return  state   from an initial state in $N^+_S$, obtained by integrating
the equations of motion. Under our general assumptions  this map is defined on almost all
initial states by Poincar\'e recurrence, and for many systems of interest it can be shown 
that $T$ is smooth on a dense open set of full measure. 
We make this almost everywhere smoothness a standing assumption.
For simplicity, we indicate the domain of $T$ simply by
$N^+_S$, ignoring the fact that it is really defined on an open dense subset of full measure. It is convenient to redefine $T$ by composing it with the reflection map $R:N^-_S\rightarrow N^+_S$,
so that $T$ becomes  a self-map of $N^+_S$.  Thus we add to the above list of assumptions:
\begin{enumerate}
\item[v.] the return map $T$ is smooth on an open dense subset of $N^+_S$ of full measure.
\end{enumerate}

\subsection{The Markov operator}

Let  $\eta$ be any given  probability measure on $ \mathbbm{T}^k\times N_{\text{\tiny wall}}$. 
The physically most natural and  interesting choice for $\eta$ corresponds to taking 
the product of the uniform distribution on $\mathbbm{T}^k$ and the Gibbs canonical distribution on $TM_{\text{\tiny wall}}$
with   parameter $\beta=1/kT$, whose definition is recalled later. The choice of measure fixes the statistical state of
the wall system.
The  collection of possible    states of the molecule system is   the space $\mathcal{P}(N_{\text{\tiny mol}})$
of  Borel  probability measures on $N_{\text{\tiny mol}}$. 
We   now define  the map $$P:\mathcal{P}(N_{\text{\tiny mol}})\rightarrow
\mathcal{P}(N_{\text{\tiny mol}})$$
that associates  to each statistical state $\mu\in \mathcal{P}(N_{\text{\tiny mol}})$   the new state $\mu P:=(\pi\circ T)_*(\mu\otimes \eta)$.
Notations and general explanations are further provided  in Section \ref{randomdynsys}.  The interpretation is that,
to obtain the return statistical state of the   molecule,  we take
its present state $\mu$, 
form the combined state $\mu\otimes \eta$ of the system, let it  evolve under $T$,  thus yielding 
$T_*(\mu\otimes \eta)$, and finally project the outcome  back to 
$N_{\text{\tiny mol}}$ under the natural projection $\pi: N^+_S\rightarrow N_{\text{\tiny mol}}$. The asterisk indicates the push-forward operation on measures.

Consider again the system of Figure \ref{exampleV} as an example.
In that case $M_{\text{\tiny wall}}$ is trivial (a single point) and $N_{\text{\tiny mol}}$ is  identified with $\{-1,1\}\times \mathbbm{H}$, where $\mathbbm{H}$ is
the half-plane in $\mathbbm{R}^2$. It does not make sense in this case to 
consider  a Gibbs canonical distribution\----a natural measure $\eta$   here   is the uniform probability distribution
on $\mathbbm{T}^1$.   Since in this example the speed of the particle does not change,  we 
consider not the full $N_{\text{\tiny mol}}$  but a level set
 for the energy (say, only states with unit velocity).
So we let  $N_{\text{\tiny mol}}=\{-1,1\}\times (-\pi/2,\pi/2)$, which parametrizes the sheet number ($\pm 1$)  and the angle $\theta$ of the incoming
trajectory relative to the normal to the wall plane.

Writing $s=(b,\theta)$ for a point in $N_{\text{\tiny mol}}$, we can define
$P$ by first indicating how it acts on, say, essentially bounded functions on $N_{\text{\tiny mol}}$ and then defining its action on $\mathcal{P}(N_{\text{\tiny mol}})$
by duality, $(\mu P)(f)=\mu(Pf)$, where $\mu(f)$ indicates the integral of $f$ with respect to $\mu$.
Thus if $f$ is a bounded function on $N_{\text{\tiny mol}}$, and $\Psi_s(x)$ is the state of return to $N_{\text{\tiny mol}}$ under $T$ for $(s, x)\in N_{\text{\tiny mol}}\times \mathbbm{T}^1$, then
from  the general  definition we have,
$$(Pf)(s)=\int_{\mathbbm{T}^1} f(\Psi_s(x))\, dx. $$

When the (macroscopic) billiard table is a channel as shown in Figure \ref{channel}, iterates of 
 $P$
give
 the post-collision states of a random flight of the molecule (say,  in a gas of non-interacting
molecules) inside the  channel.

\begin{figure}[htbp]
\begin{center}
\includegraphics[width=2.5in]{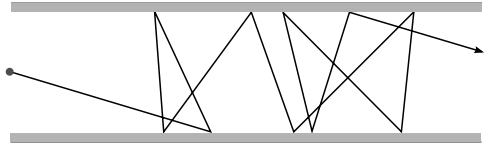}\ \ 
\caption{\small A random flight in a channel. The molecule's state after each collision is specified by
$P$ and the pre-collision state. The stationary distribution $\mu$  is interpreted as describing the state after the molecule has
reached thermal equilibrium with the wall. }
\label{channel}
\end{center}
\end{figure}

Diffusion approximation of the random flight
and the dependence of the diffusion constant on the spectrum of $P$ are issues of
particular interest, which   will be investigated   in another paper dedicated to  central limit theorems for $P$ and
related topics.

\subsection{Overview of  the  main results}
The starting point of our analysis is a determination  of the stationary probability measures of $P$.
 Recall that a probability measure  $\mu$ is said to be   {\em stationary}    for $P$ if $\mu P=\mu$.
 We   consider two possibilities: (1) the space $M_{\text{\tiny wall}}$ reduces to a point, in which case the
 wall is regarded as a rigid, unmoving  body that does not exchange energy with the billiard particle in a collision.
Note that, in this case, the billiard chamber might still have a non-trivial structure, but it does not contain moving parts.
 In this case,  the measure $\eta$ that enters in the definition of $P$ is taken to be the normalized Lebesgue measure on $\mathbbm{T}^k$;
 (2) the space $M_{\text{\tiny wall}}$ has dimension at least one. This means that the wall system has moving parts 
 and energy can be transferred between wall and molecule in a collision. In this case we assume for the purposes of the
 next theorem that $\eta$ is the product of the normalized Lebesgue measure on $\mathbbm{T}^k$ and
 the Gibbs canonical  measure on the phase space $N_{\text{\tiny wall}}$ with a fixed parameter $\beta$. 
 The latter can be written as follows (a fuller discussion of invariant measures is given in the last section of  the paper):
 \begin{equation}\label{gibbswall} d\eta_\beta = \frac{e^{-\beta E}}{Z(\beta)} \left|\Omega^E_{\text{\tiny wall}}\wedge dE\right|,\end{equation}
 where 
 $E$ is the energy function on $N_{\text{\tiny wall}}$,
 $\Omega^E_{\text{\tiny wall}}$ is the invariant (Liouville) volume form on energy level sets in  $N_{\text{\tiny wall}}$ 
 derived from the symplectic form on this   space,
  and the vertical bars indicate 
 the associated measure. The denominator is a normalization factor.
 
 We consider similar measures on $N_{\text{\tiny mol}}:=T\overline{M}_{\text{\tiny mol}}\times \mathbbm{H}.$ More precisely,
 in case (1) we fix a value $\mathcal{E}$ of the energy function of the molecule, which remains constant throughout 
 the process, and  consider the microcanonical measure for this value, given by 
 \begin{equation}\label{microcanonical}
 d\mu = \frac1{Z(\mathcal{E})} \left|\Omega^{\mathcal{E}}_{\text{\tiny mol}, S}\right|,
 \end{equation}
  where $S$ is the hypersurface of separation between the product zone and the zone of interaction previously described,
  $\Omega^{\mathcal{E}}_{\text{\tiny mol}, S}$ is the invariant   volume form on the part of the level
  set $E=\mathcal{E}$ above $S$ and the denominator is a normalizing factor. In case (2) we
  define
  \begin{equation}\label{gibbsmolecule}
  d\mu=\frac{e^{-\beta E}}{Z(\beta)} \left| \Omega^E_{\text{\tiny mol}, S}\wedge dE\right|.
  \end{equation}
  A description of these measures better  suited for applications will be given shortly.
In the special case when $\overline{M}_{\text{\tiny mol}}$ reduces to a point, so that
  $N_{\text{\tiny mol}}:= \mathbbm{H}$, these measures are as follows: in case (1) we may choose $\mathcal{E}$
  so that the molecule state lies in the unit hemisphere in $\mathbbm{H}$. Indicating by $\omega^{\text{\tiny sphere}}$
  the standard volume form on the unit hemisphere and by $\nu$ the unit normal vector to $S$, say, pointing into the
  zone of interaction, we
  have, up to a  normalization constant $C$,
  $$ d\mu(v)=C\langle v,\nu\rangle  \left| \omega^{\text{\tiny sphere}}\right|,$$
  which we refer to as the {\em Knudsen} probability distribution;
and in case (2)
  $$d\mu(v)= 
  C \langle v, \nu\rangle e^{-\beta \frac{m |v|^2}{2}} dV(v), $$
 where $\langle v,\nu\rangle$ 
 and  $dV(v)$ are, respectively, 
 the standard inner product and volume element  in $\mathbbm{R}^{k+1}$. 
 This measure is the Maxwell-Boltzmann distribution at boundary points.

\begin{theorem}\label{gibbs} 
Let 
$P:\mathcal{P}(N_{\text{\tiny mol}})\rightarrow
\mathcal{P}(N_{\text{\tiny mol}})$ be the  Markov operator associated to
a probability measure $\eta$ on $\mathbbm{T}^k\times N_{\text{\tiny wall}}$.
\begin{enumerate}
\item In case (1) above, let $\eta$ be the normalized Lebesgue measure on $\mathbbm{T}^k$. 
Then the microcanonical distribution  \ref{microcanonical} is a stationary probability for $P$.
\item In case (2) above, let $\eta$ be the product of the normalized Lebesgue measure on  $\mathbbm{T}^k$
and 
the Gibbs canonical distribution on 
$N_{\text{\tiny wall}}$ with temperature parameter $\beta$.
 Then 
the Gibbs canonical   distribution on $N_{\text{\tiny mol}}$ given by \ref{gibbsmolecule}, with  the same  parameter $\beta$,
is a stationary probability for $P$.
\end{enumerate}
In the particular case when the molecule reduces to a point, the stationary measures are, respectively, the Knudsen distribution
in case (1) and the boundary Maxwell-Boltzmann distribution in case (2).
\end{theorem}

The proof of this theorem is given at the end of Subsection \ref{productsystems}.
The stationary probability is often (but not always) unique and one often obtains
convergence of $\mu_0 P^n$ to the stationary state for any initial distribution $\mu_0\in \mathcal{P}(N_{\text{\tiny mol}})$. Thus the dynamic of
a Markov chain derived from $P$ describes the process of relaxation of the molecule's state   toward  thermal
equilibrium with the wall. To understand this process in each particular situation it is necessary to
 study the operator $P$ in more detail; we do this later in the text for a few concrete examples.

The following definition is further elaborated in Section \ref{randomdynsys}.
We say that the molecule-wall system is {\em symmetric} if   on the   state space $N^+_S$  
are defined two automorphisms $\widetilde{J}$ and $\widetilde{S}$ such that:
\begin{enumerate}
\item[i.] these maps preserve
the natural   measure $\Omega^E_{S}$ on $N^+_S(\mathcal{E})$ derived from the symplectic structure;  
\item[ii.]  they respect the 
product fibration $\pi:N^+_S=N_{\text{\tiny mol}}\times \mathbbm{T}^k\times N_{\text{\tiny wall}}\rightarrow    N_{\text{\tiny mol}}$ and both induce the same  map $J$
on $N_{\text{\tiny mol}}$; that is, 
 $$\pi\circ \widetilde{J}= \pi\circ \widetilde{S} = J\circ\pi;$$ 
 \item[iii.]  $\widetilde{J}$ is time reversing:  $\widetilde{J}\circ T=T^{-1}\circ \widetilde{J}$; and $\widetilde{S}$ commutes with
 $T$:
 $\widetilde{S}\circ T=T\circ \widetilde{S}$.  
\end{enumerate}
The existence of the map $\widetilde{J}$  is typically assured by the time reversibility of Newtonian mechanics and
the symmetry $\widetilde{S}$ can often be obtained by a simple extension of the original system that does
not affect its essential physical properties. (This is akin to defining an orientation double cover of a possibly
non-oriented  manifold.)  
The assumption of symmetry is thus a very weak one.
These points are further discussed in
Section \ref{randomdynsys} and in some of the specific examples studied later in the paper.

It is natural to consider the associated operator, still denoted $P$,
on the Hilbert space $L^2(N_{\text{\tiny mol}}, \mu)$, where $\mu$ is one of  the stationary probabilities 
obtained in Theorem \ref{gibbs}.
We are particularly interested in the spectral
theory of $P$. A first general observation in this direction is the following. 

\begin{theorem}\label{theoremselfadjoint}
Let $\mu$ be the stationary measure of $P$ obtained in Theorem \ref{gibbs} and
suppose that the system is symmetric. Then $P$ is
a   self-adjoint operator on $L^2(N_{\text{\tiny mol}}, \mu)$ of norm $1$.
\end{theorem}

 In particular, $P$ has real spectrum contained in $[-1,1]$. It is often the case (this will be proved for
 a simple but representative  example later in this paper, and has been shown for special classes of $P$ in previous papers; see \cite{F,FZ}) that $P$ is a compact, integral operator (Hilbert-Schmidt).
The eigenvalues of $P$ are then invariants of the system, depending in a canonical way on
 structural parameters like mass ratios, potential functions, curvatures, etc.
The relationship between the spectrum of $P$ and these parameters is one of the central issues
in this subject.

Of particular interest is the {\em spectral gap} of $P$, defined as $1$ minus the spectral radius of
the restriction of $P$ to the orthogonal complement to the constant functions. As is well-known (see, for example, 
\cite{RR}) the spectral gap can be used to estimate the exponential rate of convergence of $\mu_0P^n$ to
the stationary distribution in the total variation or the $L^2$ norm.

A  perturbation approach to the spectrum of $P$, which 
is valid when the molecule scattering is not far from specular, can be very fruitful. To make sense of this, first
define 
$$ \mathcal{E}_2(v):=E_v\left[\text{dist}(v,V)^2\right],$$ 
where $\text{dist}$ is a distance function on $N_{\text{\tiny mol}}$, $v\in N_{\text{\tiny mol}}$ is an initial state, 
  $V$ is the random variable 
representing the scattered state  after one collision event, and $E_v[\cdot]$ is conditional expectation given
 the initial state  $v$. We call $\mathcal{E}_2$ the {\em second moment
of scattering}. Under the identification of $N^+_S$ and $N^-_S$ (see above), specular
reflection corresponds to  $V=v$ almost surely  and 
 small deviations from specularity  correspond to  small values of the second moment of scattering. 
We now defined the operator
$$ \mathcal{L}_P:=2({P-I})/{\mathcal{E}_2},$$
which we refer to as the {\em random billiard Laplacian} (or the {\em Markov Laplacian}) of the system. 
The   billiard Laplacian, for small values of $\mathcal{E}_2$  often
  approximates   a  second order differential
operator. In the examples studied, this will be seen to be
 a (densely defined) self-adjoint operator on the same Hilbert space on which  $P$ is defined, whose eigenvalue problem 
amounts to a  
standard Sturm-Liouville    equation.

\begin{figure}[htbp]
\begin{center}
\includegraphics[width=2.0in]{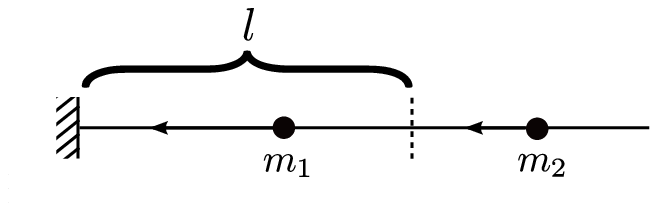}\ \ 
\caption{\small  A simple example described in Theorem \ref{casestudy}.}
\label{elementarytheorem2}
\end{center}
\end{figure}

In   Section \ref{twomasses} we explore these ideas in  detail with  an  example.
 The example consists of two point masses  (see Figure \ref{elementarytheorem2}) constrained to move along 
the half-line $[0,\infty)$. Mass $m_1$, with position coordinate $x_1$, is restricted to move in the interval $0\leq x_1\leq l$ and $m_2$,
with position coordinate $x_2$,  can move freely on $x_1\leq x_2<\infty$.
 The two masses collide elastically, and $m_1$ collides elastically with  walls at $0$ and $l$. The wall at $l$ is regarded as permeable 
 to $m_2$ but not to $m_1$.  The random  state $\eta$ of $m_1$ is taken to be the product of the uniform 
 distribution over $[0,l]$ and a Gaussian probability  with mean $0$ and variance $\sigma^2$ for its velocity.
 Mass $m_2$ is the molecule and mass $m_1$ is part of the wall system. 
 We refer to this as the {\em two-masses} system. To simplify notation  and for other conveniences
 we rescale positions and velocities according to $x:=\sqrt{m_1/m}\, x_1$ and $y=\sqrt{m_2/m}\, x_2$, where $m=m_1+m_2$.
The main structural parameter of the system is the {\em mass-ratio}  $\gamma:=\sqrt{m_2/m_1}$.
We let $P_\gamma$ represent the Markov operator with mass-ratio  $\gamma$. 
Further details are explained in  
Section \ref{twomasses}. 
We summarize in the next theorem some of  the main conclusions obtained for the two-masses example. (Further refinements
and numerical calculations are
described in that section.)
\begin{theorem}[Case study]\label{casestudy}
The following assertions hold for the two-masses system with $\gamma<1/\sqrt{3}$:
\begin{enumerate}
\item   $P_\gamma$ has a unique stationary distribution $\mu$. Its  density relative to
Lebesgue measure on $(0,\infty)$ is given by
$$ \rho(v)=\sigma^{-1}v\exp\left(-\frac{v^2}{2\sigma^2}\right).$$
\item  For an arbitrary initial probability distribution $\mu_0$, we have 
$\|\mu_0 P_\gamma^n-\mu\|_{TV}\rightarrow 0 $ exponentially fast in the total variation norm.
\item  $P_\gamma$ is a Hilbert-Schmidt operator.
\item If $\varphi$ is a function of class $C^3$ on $(0,\infty)$, then the billiard Laplacian  has the following  limit 
$$(\mathcal{L}\varphi)(v)=\lim_{\gamma\rightarrow 0} \frac{\left(P_\gamma\varphi\right)(v)- \varphi(v)}{2\gamma^2}:=  \left(\frac1{v} -v\right)\varphi'(v) + \varphi''(v).$$
Equivalently, $\mathcal{L}$ can be written in Sturm-Liouville form as $\mathcal{L}\varphi= {\rho^{-1}}\frac{d}{dv}\left(\rho\frac{d\varphi}{dv} \right),$
which  is a densely defined  self-adjoint   operator on $L^2((0,\infty),\mu)$.
\end{enumerate}
\end{theorem}

Based on part 4 of the above theorem and a simple analysis of the corresponding Sturm-Liouville eigenvalue problem
(the equation in part 4 is,  after the change of coordinates $x=v^2/2$, Laguerre's equation), we can make an educated guess as to the asymptotic value, for
small $\gamma$, of the spectral gap of $P_\gamma$: it is given by $4\gamma^2$. Although we do not 
prove this   here, we offer  in Section \ref{twomasses} numerical evidence for its validity.  This gives the following refinement of item 2 of Theorem \ref{casestudy}:
$$ \|\mu_0 P_\gamma^n-\mu\|_{TV}\leq C (1-4\overline{\gamma}^2)^n$$
$C$ is a positive constant and $\overline{\gamma}/\gamma\rightarrow 1$ as  the mass-ratio parameter $\gamma$ approaches $0$.
(A general spectral perturbation study of our Markov operators   with small moment of scattering based on   comparison with
Sturm-Liouville eigenvalue equations will be given in  another paper.)

Theorem \ref{casestudy} is  proved in Section \ref{twomasses}. Section \ref{otherexamples} discusses  further examples in less
detail.

\subsection{An algorithm for the Markov chain simulation}
The map $T$ generating the  deterministic discrete dynamical system involved in the definition of $P$ 
  is essentially a kind of  {\em billiard map} \cite{CM} in a very general setting. 
  For  applications of  Theorem \ref{gibbs}, particularly in numerical simulations of
  the Markov chains,  we need a convenient   expression for  the corresponding invariant {\em billiard measure}
  and for  the Gibbs distributions.
 We recall that the standard invariant billiard measure for
planar billiards can be described as follows: If states of the billiard system are represented by coordinates $(s, \theta)$,
where $s\in [0,1]$ is proportional to arclength measured from a reference point on the boundary of the billiard table (assuming
 that this  boundary has finite length) and $\theta\in [-\pi/2, \pi/2]$ measures the angle that a unit velocity vector
based at the  point   represented by $s$ makes with the inward pointing normal vector,
then $d\mu(s, \theta)=\frac12 \cos\theta\, ds\, d\theta$ is the canonical invariant probability  measure on the two-dimensional
state space of the billiard system. See, for example, \cite{CM}.
We wish to have a similar  description of  the invariant   billiard measure   on boundary components of   general Riemannian manifolds in the presence of  potential functions.
Although this is something rather  classical and possibly   relatively well-known, we could  not
find it  in the literature in a form that is convenient for our needs. 
It may thus be of some interest to highlight  such an expression here.
After doing this, we give at the end of this subsection the outline of an algorithm for generating Markov
chains associated to general microstructures. 

The next theorem is of some independent interest having to do with a representation of the standard volume measures of classical statistical mechanics.
Our primary interest is to apply the theorem to a neighborhood of $S$ in the  non-interaction zone
of $M$ described above.  But, in the interest of generality, the notation in this section is similar to, but independent of what was used above.  In particular, we will reuse $M$ to mean any smooth Riemannian manifold with corners and $U:M\rightarrow \mathbb{R}$ any smooth potential function.

As always, we let the kinetic energy   be defined by the Riemannian metric, $\kappa(v)=\frac12\|v\|^2$, and
consider the Hamiltonian  flow on $N=TM$ with standard energy function  $E:=\kappa+U\circ \tau$, where $\tau:N\rightarrow M$
is the base-point projection.  Let  $m$ be the dimension of $M$ and $S$ 
an $(m-1)$-dimensional submanifold of the boundary of $M$. Let
$N_S^+(\mathcal{E})$ be the intersection of $N_S^+$, defined just as  in the more specialized setting of the molecule-wall systems,
with the constant energy submanifold $E=\mathcal{E}$ of $N$.  We assume that  the return map 
 $T:N_S^+(\mathcal{E})\rightarrow N_S^+(\mathcal{E})$ is defined (in the a.e. sense described above);
let $\nu$ be the  unit  normal vector field on $S$ pointing towards the interior of $M$; and 
for any given value $\mathcal{E}$  
define $ h_\mathcal{E}:M({\mathcal{E}})\rightarrow \mathbb{R}$, where $M({\mathcal{E}}):=\{q\in M: U(q) < \mathcal{E}\}$ and 

 \begin{equation}\label{hfunction}
 h_\mathcal{E}(q):=\sqrt{2(\mathcal{E}-U(q))}.
 \end{equation}   
 Extend $h_{\mathcal{E}}$ to all of $M$ by  setting 
   $h_\mathcal{E}=0$
   on the complement  of  $M({\mathcal{E}})$. 
   Now let $W$ be an open set in $M(\mathcal{E})$ on which is defined an orthonormal frame of vector fields $e_1, \dots, e_m$; let 
   $N_W(\mathcal{E})$ be  the part of $N(\mathcal{E})$ above $W$, and 
  let $F_{\mathcal{E}}:W\times S^{m-1}\rightarrow N_W(\mathcal{E})$ be the map such that $$F_{\mathcal{E}}(q,u)=\left(q, h_{\mathcal{E}}(q)\sum_{i=1}^m u_i e_i(q)\right),$$ where
   $S^{m-1}$ is the unit sphere in $\mathbbm{R}^m$. 
 Then $F_{\mathcal{E}}$ is a diffeomorphism and  $u\mapsto F_{\mathcal{E}}(q, u)$ maps the unit sphere bijectively onto the
   fiber of $N(\mathcal{E})$ above $q$. If $W$ intersects $S$ or more generally $\partial M$, we assume that at any   $q\in W\cap \partial M$  the vector
   $e_m(q)$ is perpendicular to $T_q(\partial M)$.
  
  The Riemannian volume form on $M$ will be  denoted by $\omega^M$ and that on  $S$
  by $\omega^S$.  
  Recall that
  the relationship between $\omega^M$ and $\omega^S$ is that
$\omega^S=\nu \righthalfcup \omega^M$, the interior multiplication of $\omega^M$ by the unit normal vector vector field on  $S$.
   Let $\omega^{\text{\tiny sphere}}$ be
the Euclidean volume form    on    $S^{m-1}$, which is obtained from the standard volume form on $\mathbbm{R}^m$ by 
 interior  multiplication  with the unit radial vector field.  
In Section \ref{invariantvolumes} we define the (microcanonical) invariant forms $\Omega^E$ and $\Omega^E_S$ on $N(\mathcal{E})$
and $N^+_S(\mathcal{E})$, respectively,  in terms of the symplectic form and prove the following.

\begin{theorem}[Invariant volumes]\label{invariantvol}
For any choice of orthonormal frame over an open set  $W\subset M(\mathcal{E})$,  and given  the frame map  $F_{\mathcal{E}}:W\times S^{m-1}\rightarrow N_W(\mathcal{E})$ defined above, the
form $\Omega^E$ satisfies
$$ F^*_\mathcal{E}\Omega^E=\pm m h_\mathcal{E}^{m-2}   \omega^M\wedge \omega^{\text{\tiny sphere}}.$$
If $W$ is a neighborhood of a point in $S$, we similarly have
$$F_{\mathcal{E}}^* \Omega_S^E=\pm  \cos\theta\,  h_\mathcal{E}^{m-1}   \omega^S\wedge \omega^{\text{\tiny sphere}}, $$
where $\theta(v)$ is  the angle that $v\in T_qM$ makes 
with $\nu(q)$ for  $q\in S$. Apart from the unspecified signs, these expressions do not depend on the choice of local 
orthonormal frame.
\end{theorem}

The first volume form is invariant under the Hamiltonian flow, and the second is invariant under the return map to $S$. 
  We refer to the latter  as the {\em billiard volume form} and to the former as the {\em Liouville volume form}.  
The Gibbs canonical distribution with temperature parameter 
$\beta$ is then the probability measure obtained from the volume form
$$ e^{-\beta E}  h_\mathcal{E}^{m-2} dE\wedge  \omega^M\wedge \omega^{\text{\tiny sphere}} $$
on $N$;  similarly the Gibbs volume form is defined on $N_S^+$  using the billiard volume form just introduced above.
These volumes on $N$ and on $N^+_S$ are also invariant  under the Hamiltonian flow  and
under $T$, respectively.
The  probability measure on   $N_S^+(\mathcal{E})$ associated to the billiard volume form
may also be called the {\em Gibbs microcanonical distribution}.  
Note how the probability of occupying a state of high potential energy in the microcanonical distribution depends on
the potential function (thus on the position in $M$)
due to the term $ h_\mathcal{E}^{m-1}$. 

Given the representation of the invariant volumes of Theorem \ref{invariantvol}, we can state the Markov chain algorithm
as follows.

\begin{enumerate}
\item[MC1.] Start with a $\xi_{\text{\tiny old}}\in N_{\text{\tiny mol}}$, representing the state of the molecule prior to a collision event;
\item[MC2.] Choose   $\mathcal{E}$    with  probability density
proportional to $\exp(-\beta \mathcal{E})$, representing the energy of the wall system;
\item[MC3.]  Choose   $q\in M_{\text{\tiny wall}}(\mathcal{E})=\{q\in M_{\text{\tiny wall}}: U_{\text{\tiny wall}}(q)\leq \mathcal{E}\}$
with probability density proportional to  $h_{\mathcal{E}}^{m-2}$ relative to the Riemannian volume (which defines uniform distribution),
 where $m$ is the dimension of $M_{\text{\tiny wall}}$;
 \item[MC4.] Choose a random vector  $u$ over the unit sphere in $T_qM_{\text{\tiny wall}}$ with the uniform distribution;
 \item[MC5.] Set the  state of the wall prior to the  collision event to  $(q, h_{\mathcal{E}}(q) u)$;
 \item[MC6.] Use the combined state $(\xi_{\text{\tiny old}}, q,  h_{\mathcal{E}}(q) u)$ as the initial condition of the molecule-wall system prior to collision
 and let it evolve according to the deterministic   equations of motion until the molecule leaves the zone of interaction;  record the state
  $\xi_{\text{\tiny new}}$ of the molecule at this moment. 
\end{enumerate}

This procedure   is illustrated in
Subsection \ref{examplepotential}    with an example that is similar to that of Section \ref{twomasses} but
now involving a non-constant potential function.
 Similarly interpreting 
 Theorem \ref{gibbs},    a (typically unique) stationary distribution for a Markov chain with transitions  $\xi_{\text{\tiny old}}\mapsto \xi_{\text{\tiny new}}$
 can be sampled from in the following way:
 \begin{enumerate}
 \item[SD1.] Choose $\mathcal{E}$ with probability density  proportional to $\exp(-\beta \mathcal{E})$;
 \item[SD2.] Choose  $q$  in $\overline{M}_{\text{\tiny mol}}\times \mathbbm{T}^k$
 with probability density proportional to  $h_{\mathcal{E}}^{m-1}$, where $m$ is now the dimension of the latter manifold;
 \item[SD3.]  Choose a random vector $u$ over  the unit sphere in $T_q M_{\text{\tiny mol}}$ with probability density proportional to
 $\cos\theta$, where $\theta$ is the angle between the velocity of the molecule's center of mass and the normal to
 the submanifold $S$ (which the molecule has to cross to enter the region of interaction);
 \item[SD4.] Set the sample value of the equilibrium state of the molecule to be $(q, h_{\mathcal{E}}(q) u)$.
 \end{enumerate}

Theorem \ref{invariantvol} is proved in Subsection \ref{framevolumes}, Proposition \ref{propositionTheo4}.

\section{Random dynamical systems} \label{randomdynsys}
In this section we derive a few general facts concerning random dynamical processes
with the
  main goal  of proving  Theorem \ref{theoremselfadjoint}.   A useful perspective informing this  discussion is that 
  our Markov chains    arise  from deterministic systems of  which only partial information is accessible.   
  The notation employed below  is independent of that of the rest of the paper.

\subsection{The Markov operator in general}

Let $\pi:\mathcal{M}\rightarrow X$ denote a {\em measured fibration}, by which we simply mean a measurable map
between Borel spaces together  with a family of probability  measures $\eta=\{\eta_x: x\in X\}$ on fibers, so 
that $\eta_x(\pi^{-1}(x))=1$ for each $x$. The family is measurable in the following sense: If $f:M\rightarrow [0,\infty]$ is a Borel function then
$x\mapsto \eta_x(f)$ is Borel, where $\eta_x(f)$ indicates the integral of $f$ with respect to $\eta_x$. 
We refer to $\eta$ as the {\em probability kernel} of the fibration.

A {\em random system}  on $X$ is specified  by the   data  $(\pi, T, \eta, \mu)$, where $\pi$ is a measured fibration with
probability kernel $\eta$, $\mu$ is the {\em initial probability distribution } on $X$, and $T:\mathcal{M}\rightarrow \mathcal{M}$
is  a measurable map. We think of  the map $T$ as the generator of a deterministic dynamical system on 
the state space $\mathcal{M}$. A point $\xi$ in $\mathcal{M}$ represents a fully specified state of the system, of which
the ``observer''  can only have partial knowledge represented by $\pi(\xi)$. (It is not assumed that $T$ maps  fibers to fibers.)

From this we define a Markov chain with state space $X$ as follows. Let $\mu$ be a probability measure on $X$ representing
the statistical state of the (observable part of the)  system at a given moment.  Then the state of the system at the next
iteration is given by 
$$\mu\mapsto \mu P := (\pi\circ T)_* \mu\circ \eta. $$
The notation should be understood as follows.  From $\mu$ and $\eta$ we define a probability measure $\mu\circ \eta$ on 
$M$ so that for any, say $L^\infty$ function $f:M\rightarrow \mathbbm{R}$ 
$$(\mu\circ \eta)(f)= \int_X \eta_x(f) \ \! d\mu(x). $$ 
The   push-forward operation on measures  is defined by 
$T_*\nu(f):=\nu(f\circ T).$ The result is an operator $P$ taking probability measures to probability measures,
which we refer to as the {\em Markov operator}. 
When it is helpful to be more explicit we write, say,  $P_{\eta, T}$ or $P_\eta$, instead of $P$.

The probability kernel $\eta$ is the family of 
{\em transition probabilities} of the Markov chain.  
In keeping with standard notation, we let $P$ act on measures  (states) on the right, and on functions (observables)
on the left. Thus $P f$ is the function such that $\mu(P f)=(\mu P)(f)$ for all $\mu$.
It follows that
$$ (Pf)(x)=\int_{\pi^{-1}(x)} f(\pi\circ T(\xi))\ \! d\eta_x(\xi).$$
We say that $\eta$ is the {\em disintegration} of a probability measure $\nu$ on $\mathcal{M}$ relative to
a probability $\mu$ on $X$ if $\nu=\mu\circ \eta$. 

A probability measure $\nu$ on $\mathcal{M}$ is {\em invariant} under $T$ if $T_*\nu=\nu$, and a probability measure $\mu$ on $X$ is {\em stationary}  for the random system if $\mu P=\mu$. 
\begin{proposition}\label{invariantstationary}
Let $\nu$ be a $T$-invariant probability measure on the total space $\mathcal{M}$ of the random system $(\pi, T, \eta, \mu)$ and
 suppose that $\eta$ is the disintegration of $\nu$ with respect to $\mu:=\pi_*\nu$.  
Then $\pi_*\nu$ is a stationary probability measure on $X$.
\end{proposition}
\begin{proof}
This is immediate from the definitions:
$$ (\pi_*\nu) P =(\pi\circ T)_* (\pi_*\nu)\circ \eta=\pi_* T_* \nu =\pi_*\nu. $$
We have used that $\pi_*(\mu\circ \eta)=\mu$.
\end{proof}

Let $\mu$ be a probability measure on $X$ and define the Hilbert space $L^2(X,\mu)$ with   inner product
$$\langle f, g\rangle := \int_X f\overline{g}\ \! d\mu. $$
\begin{proposition}
Let $(\pi, T, \eta, \mu)$ be a random system, where  $T$ is an isomorphism (thus it has a measurable inverse) of the measure space $\mathcal{M}$
and   $\nu:=\mu\circ \eta$
is $T$-invariant. Let $P_{\eta, T}$ be the
associated  Markov operator.  Then $P_{\eta, T}$, regarded as an operator on $L^2(X,\mu)$, has norm $\|P_{\eta,T}\|=1$ and
its adjoint is $P^*_{\eta,T}=P_{\eta, T^{-1}}$.
\end{proposition}
\begin{proof}
Jensen's inequality implies
$$\|P_{\eta,T}f\|^2= \int_X\left| \int_{\pi^{-1}(x)} f(\pi\circ T (\xi))\ \!d\eta_x(\xi)\right|^2\ \! d\mu(x)\leq \int_X\int_{\pi^{-1}(x)} \left| f(\pi\circ T(\xi))\right|^2 \ \! d\eta_x(\xi)\ \! d\mu(x).$$
The integral on the right equals $\int_{\mathcal{M}}\left| f\right|^2\circ \pi\circ T\ \! d\nu=\int_{\mathcal{M}}\left| f\right|^2\circ \pi\ \! d\nu$,   by $T$-invariance of $\nu$. As $f\circ \pi$ is constant on fibers, this last integral is $\|f\|^2$, showing that the norm of the operator is bounded by $1$.
Taking $f=1$ shows that the norm actually equals $1$.
To see that the adjoint equals the operator associated to the inverse map, simply observe the identity
$$\int_{\mathcal{M}} f(\pi(\xi)) \overline{g}(\pi(T(\xi))) \ \! \nu(\xi) =\int_{\mathcal{M}} f(\pi(T^{-1}(\xi)) \overline{g}(\pi(\xi)) \ \! \nu(\xi),$$
which is due to $T$-invariance of $\nu$.
\end{proof}

\subsection{Time reversibility and symmetry}
Let $(\pi, T, \eta, \mu)$ be a $T$-{\em invariant  random system}, which means by definition that    $\nu=\mu\circ \eta$
is a $T$-invariant measure  so, in particular,  $\mu$ is
stationary.  
We say that the   system  is {\em time reversible}  if there is a measurable isomorphism 
 $\tilde{J}:\mathcal{M}\rightarrow\mathcal{M}$ respecting   $\pi$ and $\nu$, in the sense that it maps fibers to
 fibers and $\tilde{J}_*\nu=\nu$,  and satisfies
 $$ T\circ \tilde{J} = \tilde{J}\circ T^{-1}.$$
 Since $\tilde{J}$ respects $\pi$, it induces a measure preserving isomorphism $J:X\rightarrow X$ (for the measure $\mu$) such that
 $J\circ \pi = \pi \circ \tilde{J}.$ We denote also by $J$ the induced composition  operator on $L^2(X,\mu)$,
 so that $Jf:=f\circ J$. Notice that such  $J$ is a unitary operator on $L^2(X,\mu)$. 
We call $\tilde{J}$ the {\em time-reversing map} of the system. 

\begin{proposition}\label{Jadj}
Let $(\pi, T, \eta, \mu)$ be a $T$-invariant random system with time-reversing map $\tilde{J}$, and $J$ its associated
unitary operator on $L^2(X,\mu)$. Then 
$ P_{\eta, T}^*= J^* P_{\eta,T}J.$
\end{proposition}
\begin{proof}
 A  straightforward consequence of the definitions is  that
 $$\left(P_{\eta,T}Jf\right)(x)=\int_{\pi^{-1}}(f\circ\pi)\left(T^{-1}\circ\tilde{J}(\xi)\right)\ \! d\eta_x(\xi), $$
from which we obtain
$$ \langle  P_{\eta,T}Jf, Jg\rangle = \int_{\mathcal{M}}(f\circ \pi)\left(T^{-1}\circ\tilde{J}(\xi)\right) \overline{g}\left(\pi\circ \tilde{J}(\xi)\right)\ \!d\nu(\xi).$$
The last integral is now seen  to be equal 
to $ \int_M f(\pi(\xi))\overline{g}(\pi(T(\xi)))\ \! d\nu(\xi)=\langle f, P_{\eta, T}g\rangle$ 
by using the invariance of $\nu$ under $\tilde{J}$ and $T$.
\end{proof}

We say that $\tilde{S}:\mathcal{M}\rightarrow \mathcal{M}$ is an {\em automorphism}, or a {\em symmetry} of the random system 
 if
it is a  measurable isomorphism commuting with $T$ that  respects    $\pi$ and  $\nu=\mu\circ \eta$.
 Thus $\tilde{S}$ covers a measure preserving isomorphism of $X$, which we denote by $S$. 
 
 \begin{definition}
The $T$-invariant, time reversible  random system $(\pi, T, \eta, \mu)$ with time reversing map $\tilde{J}$
will be called {\em symmetric} if there exists an automorphism $\tilde{S}$ whose induced map $S$ on $X$ coincides
with the map $J$ induced from $\tilde{J}$.
\end{definition}

\begin{proposition}\label{selfadjoint}
Let  $(\pi, T, \eta, \mu)$ be a symmetric (hence time-reversible and $T$-invariant) random system. Then the Markov
operator $P_{\eta, T}$ is self-adjoint.  In particular, 
$P_{\eta, T^{-1}}=P_{\eta, T}$.
\end{proposition}
\begin{proof}
Given Proposition \ref{Jadj}, it is enough to verify that $P_{\eta, T}$ commutes with the   operator $J$. Keeping in mind that $J=S$, $T\circ \tilde{S}=\tilde{S}\circ T$, 
and that $\nu$ is $\tilde{S}$-invariant, we obtain
\begin{align*}
\int_{\mathcal{M}} f\left(J\circ\pi\circ T(\xi)\right) \overline{g}\left(\pi(\xi)\right) d\nu(\xi)&=\int_{\mathcal{M}} f\left( \pi\circ \tilde{S}\circ T(\xi) \right)\overline{g}\left(\pi(\xi)\right)  d\nu(\xi)\\
&=\int_{\mathcal{M}} f\left( \pi\circ T(\xi)\right) \overline{g}\left(\pi\circ \tilde{S}^{-1}(\xi)\right)  d\nu(\xi)\\
&=\int_{\mathcal{M}} f\left( \pi\circ T(\xi)\right) \overline{g}\left(J^{-1}\circ\pi(\xi)\right)  d\nu(\xi).
\end{align*}
This means that $\langle P_{\eta, T} Jf, g\rangle=\langle P_{\eta, T}f, J^{-1}g\rangle$. The claim   follows  as  $J$ is unitary.
\end{proof}
\begin{corollary}\label{corollarySA}
Let $(\pi, T, \eta, \mu)$ be a $T$-invariant random system, where $T$ is an involution, i.e., $T^2$ equals the identity map on $\mathcal{M}$.
Then the Markov operator $P_{\eta, T}$ is self-adjoint.
\end{corollary}
\begin{proof}
Since $T=T^{-1}$,  the identity map on $I$ is both a time reversing map and a symmetry of the system. 
\end{proof}

The situation indicated in Corollary \ref{corollarySA} essentially describes a general stationary Markov
chain satisfying  {\em detailed balance.}  In that case, $X$ is the state space of the Markov chain,
$\mathcal{M}= X\times X$ (or, more generally, a measurable equivalence relation on $X$), and $\eta_x$ is a probability measure on $X$ for each $x\in X$. For $T$ we take the (groupoid) inverse map $(x,y)\mapsto (y,x)$.
As a special case,  we suppose that $X$ is countable and write in more standard notation  $p_{xy}:=\eta_x(\{y\})$ and $\pi_x:=\mu(\{x\})$.
Then the random system is symmetric as define above exactly when the Markov chain on $X$ with transition probabilities $(p_{xy})$
and stationary distribution $(\pi_x)$ satisfies the detailed balance condition
$\pi_x p_{xy}=\pi_yp_{yx} $
 for all $x, y\in X$. 
 
\subsection{Quotients}\label{quotients}
A minor  technical issue to be mentioned later calls for a consideration of {\em quotient random systems}.
Suppose that $G$ is a group of symmetries of the random system $(\pi, T, \eta, \mu)$
and that the action of $G$ on $\mathcal{M}$, as well as the induced action on $X$, are {\em nice} in the
sense that the quotient measurable spaces are  countably separated. (See, for example, \cite{zimmer}, where such actions
are called {\em smooth}. In the situations of main  concern to us, $G$ is a finite  group acting by homeomorphisms of
a metric space, in which case the condition holds.)  Without further mention let the actions of 
$G$ be nice.

 We denote the quotient system by $(\overline{\pi}, \overline{T}, \overline{\eta}, \overline{\mu})$,
 and the other associated notions, such as $\overline{\mathcal{M}}$ and $\overline{\nu}$, are
 similarly indicated with an over-bar.
 The quotient maps $\mathcal{M}\rightarrow \overline{\mathcal{M}}$  and  $X\rightarrow \overline{X}$
 will both be indicated by $p$.
  These   maps and measures are all  defined in the most natural way. 
 For example, $\overline{T}$ is the transformation
 on $\overline{M}$ such that $\overline{T}\circ p = p\circ T$, which exists since $G$ commutes with $T$,
$\overline{\mu}=p_*\mu$, etc.
 Since $G$ leaves $\mu$ invariant, it is represented  by unitary transformations  of $L^2(X, \mu)$ 
 commuting  with $P_{\eta, T}$.
 So  it makes sense to restrict $P_{\eta,T}$ to the closed subspace $L_G^2(X, \mu)$ of $G$-invariant vectors,
 which is isomorphic to the Hilbert space $L^2(\overline{X}, \overline{\mu})$ under the
 map $p^*: L^2(\overline{X}, \overline{\mu})\rightarrow L^2_G(X, \mu)$ that sends $f$ to $f\circ p$.
Thus we may identify the quotient Markov operator  $P_{\overline{\eta}, \overline{T}}$ with the restriction of $P_{\eta, T}$ to
the $G$-invariant subspace.  

We use these remarks later in situations where 
the initial system does  not have any symmetries
that would allow us to apply
 Proposition \ref{selfadjoint}, but it can nevertheless    be regarded as the  quotient of
another system  having the necessary symmetry.

\section{A case study: the two-masses system}\label{twomasses}
In this section we prove 
Theorem \ref{casestudy} and further refinements,  which  are concerned with  one simple but illustrative molecule-wall system.
Other examples
are described more briefly later in the next section.

\subsection{Description of the example}
Let $x_1, x_2  $ represent  the positions of two point masses restricted to move on the half-line $[0,\infty)$, with respective masses $m_1$ and $m_2$.
Mass $m_1$, called the {\em bound mass}, is restricted  to move between $0$ and $l$, and it reflects elastically upon hitting $0$ or $l$, moving freely between
these two values. Mass $m_2$, called the {\em free mass},  moves freely on the interval $[x_1, \infty)$, and collides elastically with $m_1$ when $x_2=x_1$. There are no forces or interactions of  any kind  other than   collisions.
We are interested in the following  ``scattering  experiment'': Mass $m_2$  starts at some place along the half-line with coordinate $x_2>l$
and speed $s>0$, moving  towards the origin. It eventually collides with $m_1$, possibly several times, before reversing direction and
leaving the ``interval of interaction'' $[0,l]$. When it finally reaches again a point with coordinate $x_2>l$, we register its new speed $S>0$ (now
moving away from the origin). 
It is assumed that we can measure $s$ and $S$ exactly, but that the state of $m_1$ at the moment the incoming mass
crosses the boundary $x_2=l$ is only known up to a probability distribution. 
\begin{figure}[htbp]
\begin{center}
\includegraphics[width=3in]{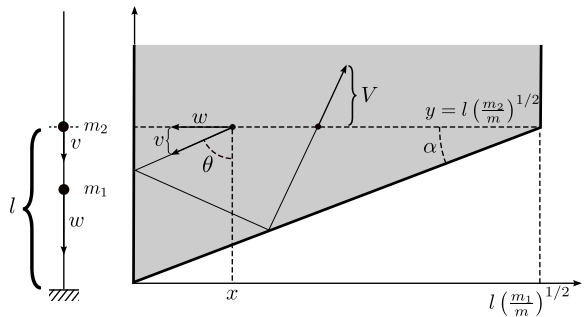}\ \ 
\caption{\small  A two-dimensional billiard system corresponding to the one-dimensional scattering example
described in the text. For the random system, start with a $v\in (0, \infty)$, then choose a random (uniform) position along the horizontal dashed line and a random value for $w$ according to a 
given, fixed, probability distribution on $\mathbbm{R}$.  Then execute  the billiard particle motion inside the triangle,
with that initial position and initial velocity $ we_1-ve_2$, where $e_i$ are the standard basis of  $\mathbb{R}^2$. Eventually, the billiard particle
returns to the horizontal dashed line. When it does, let $V$ be the absolute value of the vertical component of the   particle's velocity.
Then $v\mapsto V$ is the random map   for the  free particle  scattered velocity.}
\label{elementarypict}
\end{center}
\end{figure} 

More specifically, we assume that, at that moment  in time,
$x_1$ is a random variable uniformly distributed over $[0,l]$ and that   the velocity $v_1\in \mathbbm{R}$ of the bound mass  has a known probability distribution
that we do not yet specify. 
It is  imagined that  $l$ is  very small (``microscopic'') and  that  the bound mass is  part of the wall,
whose precise dynamical state is thus  only known probabilistically.
We wish to investigate the random process $s\mapsto S$, which gives the random  speed of the 
free particle after one  (macro-) collision  with the wall (consisting  of possibly many  collisions with $m_1$), given
its speed before the collision.

 The   state of the
system at a moment when $x_2=l$ is fully described by the triple $(x_1, v_1, s)$. It will be convenient to
use coordinates $x:= \sqrt{m_1/m}\, x_1$, $y=\sqrt{m_2/m}\,  x_2$, $w= \sqrt{m_1/m}\, v_1$ and $v=\sqrt{m_2/m}\, s$,
where $m=m_1+m_2$.  A configuration of this two-particle system
corresponds to a point $(x,y)$ in the billiard table  region depicted in Figure \ref{elementarypict}.
By doing this  coordinate change  
 we now have  ordinary  billiard motion, i.e., uniform rectilinear motion between collisions 
and 
  specular collisions at the boundary segments. (The coordinates change turns the kinetic energy of the system  into the norm, up to uniform scaling,
associated to
 the standard  inner product in $\mathbbm{R}^2 $.)

States of the two-particle system are represented by tangent vectors on the billiard region. 
Expressing the situation in the language of Section \ref{randomdynsys},
let  $\mathcal{M}$ be
the set of states for which $y=l\sqrt{m_2/m}$ (see Figure \ref{elementarypict});  omitting the $y$-coordinate, we
write
$\mathcal{M}:=\left\{\left(x,   w, v\right): x\in \left[0,l\sqrt{m_1/m}\right], v>0\right\}.$
The observable states are represented 
by  $X=(0,\infty)$ and we let  $\pi:\mathcal{M}\rightarrow X$ be the projection on the third coordinate.
The transformation $T:\mathcal{M}\rightarrow\mathcal{M}$ is the billiard return map to the horizontal dashed line of
Figure \ref{elementarypict}.
We choose 
the measure $\eta_v$ on the fiber of  $v\in X$ to be the  product $\lambda\otimes \zeta$ of the normalized Lebesgue measure on 
$ \left[0,l\sqrt{m_1/m}\right]$ and a fixed probability measure $\zeta$ on $\mathbbm{R}$.

For each choice of $\zeta$, we obtain a Markov operator $P$  of   a random process with state space $(0,\infty)$.
Such a process may be  interpreted as a sequence of successive collisions of the free mass with the ``microstructured wall.''
It may be imagined that the free particle actually moves in a finite interval of arbitrary length bounded by two such walls having  the same
probabilistic description, so that the process defined by $P$ describes the evolution of the random velocity of mass
$m_2$ as it collides alternately with the left and right walls.

\subsection{Stationary probability distributions}
For concreteness, let us choose the measure $\zeta$ to be the absolutely continuous probability measure on $\mathbbm{R}$
with Gaussian density $\rho_{\text{\tiny wall}}(w)=(\sigma \sqrt{2\pi})^{-1}\exp\left(-\frac12 w^2/\sigma^2\right)$.
This amounts  to assuming that the state  of the bound mass satisfies the  Gibbs canonical distribution with
temperature proportional to $\sigma^2$. One should keep in mind that by the above  change of coordinates 
the kinetic energy of the bound mass is simply $\frac12 \omega^2$, so
it makes sense to refer to $\sigma^2$ as the {\em temperature} of the wall.

\begin{proposition}\label{MBstationary}
Let $P$ be the Markov operator with state space $(0,\infty)$  for the random process    of the two-masses system.
We assume that the state of the  bound mass (or  the wall system) has the Gibbs canonical distribution   $\lambda\otimes \zeta$
with wall temperature $\sigma^2$. Then $P$ has a unique stationary distribution $\mu$, whose density
relative to the Lebesgue measure on $(0,\infty)$ is the Maxwell-Boltzmann distribution with the same temperature:
$$ \rho_{\text{\tiny free}}(v)={\sigma^{-2}}v\exp\left(-\frac{ v^2}{2\sigma^2}\right).$$
If $\eta$ is any Borel probability measure on $(0,\infty)$, then $\eta P^n$ converges in the weak-* topology
to the stationary probability measure.
\end{proposition}
\begin{proof}
We begin by showing that the Maxwell-Boltzmann distribution is indeed stationary for $P$ without invoking the more general
Theorem \ref{gibbs}.  
Let $\mu$ denote the probability measure with density $\rho_{\text{\tiny free}}$. Then the probability measure
$\nu=\mu\circ \eta$ on $\mathcal{M}$ is  $\mu\otimes \zeta \otimes \lambda$, which has density
$$\rho(x, v, w)=C v \exp\left(-\frac{v^2 + w^2}{2\sigma^2}\right) dx\, dv\, dw,$$
where $C$ is a normalization constant.  Let $\theta$ be the angle in $[-\pi/2,\pi/2]$ that the initial velocity
of the billiard particle in Figure \ref{elementarypict} makes with the normal vector to the horizontal dashed line pointing
downward and $r$ the Euclidean norm of the velocity vector. 
Then  $\nu$ can be written, in polar coordinates, as
 $d\nu(x, \theta, r)=C r^2 \exp\left(-\frac{r^2}{2\sigma^2}\right) \cos\theta \, dx \, d\theta \, dr$
 where $C$ is now a different normalization constant.  For each value of $r$, the   measure
 $\cos\theta \, dx \, d\theta$ is     invariant   under  the billiard map restricted
 to a constant energy surface  (see \cite{CM}; also compare with our more general Theorem \ref{invariantvol}). In our case, it is invariant under the 
   return billiard map $T$ restricted to each coordinate  $r$-slice. 
 Therefore, $v$ is itself $T$-invariant. We can then apply Proposition \ref{invariantstationary} to conclude that $\mu$ is $P$-stationary.
It will be shown shortly that $P$ is an integral operator of the form $(Pf)(v)=\int_0^\infty \kappa(v,u)f(u)\, du$, where $\kappa(v,u)>0$ for
each $v$ and all $u$. In particular, it is indecomposable  and non-periodic, and for each $v$, $\delta_v P$ is
absolutely continuous with respect to the Lebesgue measure on the half-line. By, say, Theorem 7.18 of 
\cite{breiman},  the stationary measure is unique and the claimed convergence holds.
 \end{proof}

\begin{figure}[htbp]
\begin{center}
\includegraphics[width=3.0in]{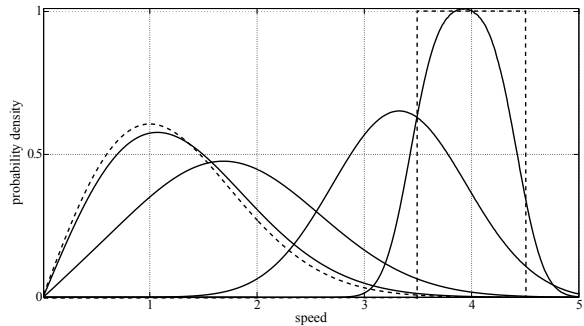}\ \ 
\caption{\small Evolution of an initial probability measure, $\mu_0$,  of the free mass velocity,  having a step function density.  
The graphs in dashed line are the initial and the limit density $v\exp(-v^2/2)$, and the other graphs, from right to left, are the densities
of $\mu_0 P^n$
at steps $n=1, 10, 50, 100$, of the Markov chain of Proposition \ref{MBstationary}.
We have used a finite rank approximation of $P$ obtained by numerically simulating the mechanical system
with mass ratio $m_1/m_2=100$.  }
\label{evolution}
\end{center}
\end{figure}

Reverting to the non-scaled variables and 
introducing the parameter $\beta$ such that $\beta^{-1}=m_1\sigma_1^2$, where $\sigma_1^2$ is the variance of the velocity 
of the bound mass $m_1$, then the stationary distribution for the speed of the free mass has density
$$\rho_{\text{\tiny MB}}(v_2)={\beta m_2}v_2 \exp\left(-\beta \frac{m_2 v_2^2}{2}\right).$$ 
We interpret $\beta$ as the reciprocal of the wall temperature.

\subsection{The random map}
Proposition \ref{MBstationary} gives the equilibrium (stationary) state of the free mass velocity process. 
This equilibrium state is arrived at by iterating a random map on $(0,\infty)$ with transition probabilities operator
$P$. We wish now to describe this  random map more explicitly.  In the following analysis, we assume that
$m_1>3m_2$.
First we set some notation: Let $\gamma:=\sqrt{m_2/m_1}=\tan \alpha$, where $\alpha$ is the 
angle of the billiard table triangle indicated on Figure \ref{elementarypict}.  Define 
$$a:=\frac{1-\gamma^2}{1+\gamma^2}, \ \ b:=\frac{2\gamma}{1+\gamma^2},\  \ \overline{a}:=\frac{1-6\gamma^2 +\gamma^4}{(1+\gamma^2)^2}, \ \
\overline{b}:=\frac{4\gamma(1-\gamma^2)}{(1+\gamma^2)^2}.$$ 
Also define the functions
$$p_v(w):= \frac{\gamma}{\sqrt{1+\gamma^2}}\frac{|w|}{v}, \ \ q_v(w):= \frac{2(1-\gamma^2)}{1+\gamma^2} -\frac{4\gamma}{1+\gamma^2}\frac{|w|}{v},$$
and introduce the partition of $(0,\infty)$ into intervals $I^i_w=|w| I^i$, $i=1,2,3,4$, where
$$I_1:=(0,\tan\alpha],\ \ I_2:=(\tan\alpha, \tan(2\alpha)], \ \ I_3:=(\tan(2\alpha), \tan(3\alpha)], \ \ I_4:=(\tan(3\alpha),\infty).$$
It is  useful to note that  $\tan(2\alpha)=2\gamma/(1-\gamma^2)$ and $\tan(3\alpha)=\gamma(3-\gamma^2)/(1-3\gamma^2)$.
To simplify the description of the map, we make the assumption that $m_1>3 m_2$, which is equivalent to $\alpha<\pi/6$.
The random map can now be expressed  as follows: Choose $w\in \mathbbm{R}$ at random with probability $\zeta$ (say, the Gaussian
probability with temperature $\sigma^2$) and define the
affine maps 
$$F^w_1(v):=av+bw,\ \ F^w_2(v):=av - bw, \ \ F^w_3(v):=-\overline{a}v + \overline{b}w. $$
These are the deterministic branches of the random map (see Figure \ref{randommap}).

\vspace{0.1in}
\begin{figure}[htbp]
\begin{center}
\includegraphics[width=4in]{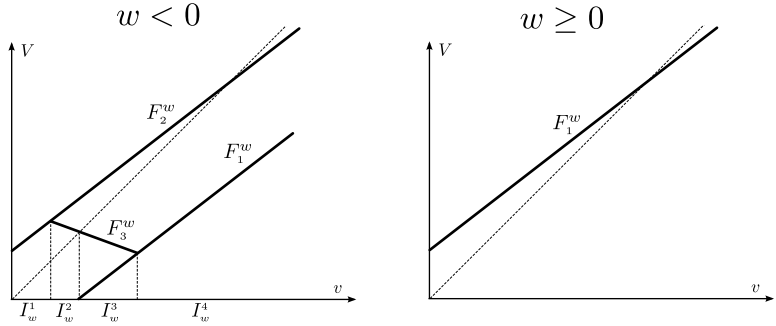}\ \ 
\caption{\small Graph of the random map $F$. The dashed line is the graph of the identity map.}
\label{randommap}
\end{center}
\end{figure} 

Finally, let $F^w:(0, \infty)\rightarrow (0,\infty)$ be the  piecewise affine random map defined 
on each interval $I^i_w$ of the partition as follows. 
Case I: If $w\geq 0$, then $F^w(v)=F^w_1(v)$.  Case II: If $w<0$, then
$$ \begin{array}{ll}\left.F^w\right|_{I^1_w}(v):= F^{|w|}_1(v)   &
 \left.F^w\right|_{I^2_w}(v):=\begin{cases} F^{|w|}_1(v) & \text{w. prob. } p_v(w)\\
F^{|w|}_3(v) & \text{w. prob. } 1-p_v(w) \end{cases} \\
\left.F^w\right|_{I^3_w}(v):= 
\begin{cases}F^{|w|}_1(v) & \text{w. prob. } p_v(w)  \\  F^{|w|}_2(v) & \text{w. prob. } q_v(w)
 \\
F^{|w|}_3(v) & \text{w. prob. } 1- q_v(w)-p_v(w) 
\end{cases} &
\left.F^w\right|_{I^4_w}(v):=
\begin{cases}F^{|w|}_1(v) & \text{w. prob. } p_v(w) \\  F^{|w|}_2(v) & \text{w. prob. } 1-p_v(w)
 \end{cases} 
\end{array}$$
(`w. prob.' $=$ `with probability.')
These expressions are  obtained   using the standard idea of ``unfolding'' the polygonal billiard
table and some tedious but straightforward work.

\subsection{A remark about symmetry}
Before  continuing  with the analysis of the example, let us briefly  examine a small modification of it
to illustrate a general  point concerning symmetries.

The modified example
 is shown  in Figure \ref{elementaryII}.   
By Proposition \ref{selfadjoint}, its Markov operator  is self-adjoint on $L^2(\mathbbm{R}, \mu_{\text{\tiny sym}})$, where
$ \mu_{\text{\tiny sym}}$ is the even measure on the real line whose conditional probability distribution, conditional on the event
that the free mass approaches from the right-hand side of the wall, 
 equals the stationary measure asserted in Proposition \ref{MBstationary}. Denoting this operator by  $P_{\text{\tiny sym}}$,
then 
$P_{\text{\tiny sym}}=( P + P^*)/2,$ where $P$ is the Markov operator of the first version of the example.

\begin{figure}[htbp]
\begin{center}
\includegraphics[width=3in]{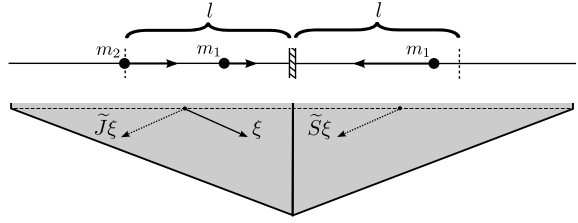}\ \ 
\caption{\small   A symmetric version of the example of Figure \ref{elementarypict} (top of figure). The free particle may approach
the wall from the left or from the right, with equal probabilities. The   
states of the bound masses on the right and   left are equally distributed and independent.  The map $\widetilde{J}$ is the time reversing map,
and $\widetilde{S}$ is a symmetry compatible with $\widetilde{J}$, in the sense that Proposition \ref{selfadjoint} applies.}
\label{elementaryII}
\end{center}
\end{figure} 

Although this remark illustrates a useful general point made in Subsection \ref{quotients}, it turns out that
$P$ is already self-adjoint in the present case. In fact, the map $J$ on $X$ such that $\pi\circ\widetilde{J}=J\circ \pi$
is the identity so the symmetry  $\widetilde{S}$ can be taken to be the identity map itself and  
Proposition \ref{selfadjoint} applies to the original system.

\subsection{The integral kernel and  compactness}
We now wish to show that the operator $P$ on $L^2((0,\infty),\mu)$ is Hilbert-Schmidt. This is the content of
the below Proposition \ref{HilbertSchmidtprop}.

It is not difficult to show that $P$ is an integral operator,  
 $(Pf)(v)=:\int_0^\infty \kappa(v, u) f(u) \, du$, whose integral kernel $\kappa$ has the following description. 
Write  $\rho_\sigma(w):=\rho_{\text{\tiny wall}}(w)$   to emphasize the parameter $\sigma$ of the Gaussian distribution. 
Then,  
$$\kappa(v,u)=\rho_{b\sigma}(u-av)\sum_{i=1}^5 Q_i\left(v,\frac{u-av}{b}\right)\mathbbm{1}_{J_i(v)}(u) + \rho_{\overline{b}\sigma}(u+\overline{a}v)\sum_{i=1}^2\overline{Q}_i\left(v, \frac{u+\overline{a}v}{\overline{b}}\right)\mathbbm{1}_{\bar{J}_i}(u)$$
where $\mathbbm{1}_A$ is the indicator function of a set $A$,  the   $Q_i$  are  
$$ Q_1(v,w)=Q_2(v,w)=1, \ \ Q_3(v,w)={p_v(w)}, \ \ Q_4(v,w)= {q_v(w)}, \ \ Q_5(v,w)=1-p_v(w),$$
the $\overline{Q}_i$ are
$$\overline{Q}_1(v,w)={1-p_v(w)}, \ \  \overline{Q}_2(v,w)={1-p_v(w)-q_v(w)}, $$
the  intervals
$J_i$ and $\bar{J}_i$ are
$$J_1=(av,\infty),    J_2=(\overline{c}v,\infty),   J_3=(av, \overline{c}v),  J_4=(0, v/\overline{c}),  J_5=(v/\overline{c}, av),  
\bar{J}_1=(v,\overline{c} v),    \bar{J}_2= (\underline{c}v, v)$$
and, finally,  $$\underline{c}=\frac{\gamma^2(3+3\gamma^2+\gamma^4)}{(1+\gamma^2)(3-\gamma^2)}, \ \ \overline{c}=\frac{3-\gamma^2}{1+\gamma^2}.$$

Let $\mu$ be the measure on $(0,\infty)$ having density  $\rho_{\text{\tiny free}}(v)=\sigma^{-2} v \exp\left(-v^2/2\sigma^2\right)$
and   $K(v,u)$  the integral kernel of $P$ relative to $\mu$.
 Thus $(Pf)(v)=\int_0^\infty K(v,u)f(u) \, d\mu(u)$.  
Then there exists a constant $C$ such that
\begin{equation}\label{kernelbound}
K_0(v,u):=\rho_{b\sigma}(u-av)/\rho_{\text{\tiny free}}(u)=\frac{C}{u}\exp\left\{-\frac{1}{2\sigma^2}\left(\frac{u-av}{b}\right)^2 + \frac{u^2}{2\sigma^2}\right\}. \end{equation}
Define $\overline{K}_0(v,u)$ similarly, by substituting $\overline{a}$ and $\overline{b}$ for $a$ and $b$.

\begin{proposition}\label{HilbertSchmidtprop}
Let $\widetilde{K}(v,u)$ be one of the following kernels: 
$$K_0(v,u) \mathbbm{1}_{(cv,\infty)}(u),  \ \ \overline{K}_0(v,u) \mathbbm{1}_{(cv,\infty)}(u), \  \text{ or }\ K_0(v,u) \, q_v\left(\frac{u-av}{b}\right)\mathbbm{1}_{(0,v/\overline{c})}(u)$$
where $c$ is a positive constant.
Then $\widetilde{K}$ has finite Hilbert-Schmidt norm with respect to the measure $\mu$ on $(0,\infty)$.
It follows that $P$ is a Hilbert-Schmidt  self-adjoint operator.
\end{proposition}
\begin{proof}
This amounts to showing that $\int_0^\infty \int_0^\infty K(v,u)^2 \, d\mu(v)\, d\mu(u)<\infty.$
Expressing the integrand in terms of  the  Lebesgue measure $dv\, du$,
omitting multiplicative constants, and  setting $\sigma=1$ for simplicity,
we have to show,  for the first of the three kernels, that
$$I:= \int_0^\infty \int_0^\infty  \exp\left\{-\left(\frac{u-av}{b}\right)^2 +  {u^2} -\frac{u^2 +v^2}{2}\right\}\frac{v}{u} \mathbbm{1}_{(cv,\infty)}(u)  \, dv\, du<\infty.$$
Making the substitution $s=v/u$, we obtain
$$I=  \int_0^\infty u \int_0^{1/c}  \exp\left\{-u^2\left(\left(\frac{1-as}{b}\right)^2  +\frac{s^2 -1}{2}\right)\right\}s  \, ds\, du,$$
which clearly is finite.  The second kernel is treated in the same manner.  
 To deal with the third kernel, first
observe  that 
$$f(s):=\left(\frac{1-as}{b}\right)^2  +\frac{s^2 -1}{2}\geq \epsilon:=\frac12{(1-a^4)}/{(1+a^2)^2} >0$$
for all $s$. We have used    $a^2+b^2=1$ and $a=(1-\gamma^2)/(1+\gamma^2)<1$.  The same change of variables, $s=v/u$,
now  yields
$$I= 2\int_0^\infty u\exp\left(-\epsilon u^2\right) \int_{\overline{c}}^\infty  \exp\left\{-u^2\left(\left(\frac{1-as}{b}\right)^2  +\frac{s^2 -1}{2} -\epsilon\right)\right\}  \, ds\, du $$
where the identity $q_v((u-av)/b)=2/s$ for  $u<v/\overline{c}$ was used.
The value of the  integral in $s$ decreases in $u$ as $O(1/u)$, so $I$ converges. The actual kernel of $P$ is a union of kernels of the
three types considered, so it also has finite Hilbert-Schmidt norm. It is interesting to note that the norm worsens as $\gamma$ approaches $0$.
\end{proof}

It is clear from the description of the integral kernel of $P$ (and by using  the general results and definitions from, say
\cite{MT}) that Markov chains associated to $P$ are Lebesgue-irreducible, strongly aperiodic, recurrent, and admit a unique
stationary probability.

\subsection{A perturbation approach to the spectrum of $P$}
For convenience, we make the velocity variables dimensionless by dividing them
by $\sigma$, which is the standard deviation of the wall-mass velocity. In this way, the stationary probability distribution
for the free mass becomes $\rho(z)= z \exp(-z^2/2)$, where $z=v/\sigma$. The specific form of the (dimensionless) velocity distribution
of the wall mass will be unimportant, but we assume that it has mean zero and variance $1$.
Having fixed the temperature, the  main  parameter of   interest is $\gamma=\sqrt{m_2/m_1}$. 
We denote by $P_\gamma$ the Markov operator   acting on   $L^2((0,\infty),\mu)$, where
$d\mu(z)=\rho(z)\, dz$.
Then
the main remark of this subsection is the following proposition.

\vspace{0.1in}
\begin{figure}[htbp]
\begin{center}
\includegraphics[width=4in]{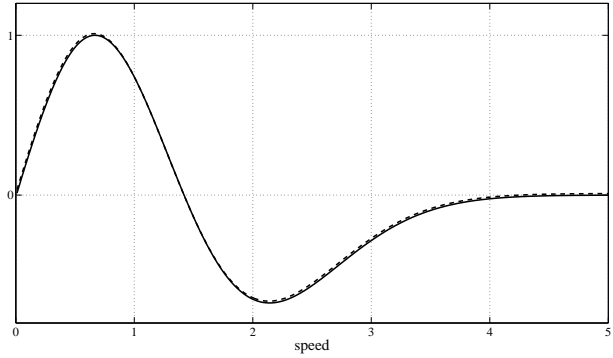}\ \ 
\caption{\small Comparison of the second eigendensity of $P$, obtained by numerical approximation,
and the second eigendensity of the billiard Laplacian $\mathcal{L}$, which is $(1-z^2/2)\rho(z)$, both normalized
so as to have maximum value $1$. The graph of the latter (dashed line) has been intentionally  offset upwards by a small
amount to better distinguish the
two. We have used $\gamma=0.1$; the numerical value for the second eigenvalue of $P$ (after $\lambda=1$) was found to
be $0.9606$, to be compared with $1+2\gamma^2 (-2)= 0.9600$, which uses the eigenvalue $-2$ of $\mathcal{L}.$}
\label{MBsecond}
\end{center}
\end{figure}

\begin{proposition} 
If  $\phi$ is  a function of class $C^3$ on $(0,\infty)$ vanishing at $0$ and $\infty$, then
$$\lim_{\gamma\rightarrow 0} \frac{\left(P_\gamma\varphi\right)(z)- \varphi(z)}{2\gamma^2}=(\mathcal{L}\varphi)(z)$$
holds for all $z>0$ where $\mathcal{L}$,
the  {\em billiard Laplacian} of the system of Example \ref{elementarypict}, is defined  by
$$(\mathcal{L}\phi)(z):=  \left(\frac1{z} -z\right)\varphi'(z) + \varphi''(z).$$
Equivalently, $\mathcal{L}$ can be written in Sturm-Liouville form as $\mathcal{L}\varphi= {\rho^{-1}}\frac{d}{dz}\left(\rho\frac{d\varphi}{dz} \right) .$
This is a densely defined, self-adjoint operator on $L^2((0,\infty),\mu)$.
\end{proposition}
\begin{proof}
For the sake of brevity, we give the basic  idea of the proof in the special when the velocity distribution of
the wall-bound mass is Bernoulli, taking  values $\pm1$ with equal probabilities $1/2$. We make the additional simplification
of ignoring the branch $F_3^w$ of the random map (see Figure \ref{randommap}). Notice that the intervals $I^1_w, I^2_w, I^3_w$
(same figure) lie in $[0,3\gamma+O(\gamma^2)]$, so only $F_1^w$ and $F_2^w$ are expected to be important. 
Thus we consider the simpler random system define as follows. Let $F_1(z):=az+b$ and $F_2(z):az-b$, where 
we approximate $a=1-2\gamma^2$ and $b=2\gamma$. Let $p:=\gamma/z$.
Then the approximate random dynamics corresponds to applying $F_1$ with probability $(1+p)/2$ and $F_2$ with
probability $(1-p)/2$.
Define the {$k$-th moment of scattering} as 
$$\mathcal{E}_k(z):=E_z[(Z-z)^k]=\frac{1+\gamma/2}{2}\left(F_1(z)-z\right)^k+\frac{1-\gamma/2}{2}\left(F_2(z)-z\right)^k, $$
where $Z$ is the random speed after collision, $z$ is the speed before collision, and $E_z[\cdot]$ indicates expectation given $z$. 
From this general expression we derive
$$ \mathcal{E}_1(z)= 2\gamma^2 \left(\frac1z -z\right),\ \ \mathcal{E}_2(z)= 4\gamma^2 +O(\gamma^4), \ \ \mathcal{E}_k(z)=O(\gamma^k).$$
Now, we expand
$(P_\gamma\varphi-\varphi)(z)/{2\gamma^2}=E_z\left[\left({\varphi(Z)-\varphi(z)}\right)/{2\gamma^2}\right]$
in Taylor polynomial approximation up to degree $2$ and obtain
$$
\frac{(P_\gamma\varphi-\varphi)(z)}{2\gamma^2}
=\frac1{2\gamma^2}\left(\varphi'(z)E_z\left[Z-z\right] + \frac12 \varphi''(z)E_z\left[(Z-z)^2\right] +E_z\left[O\left(|Z-z|^3\right)\right]\right)
=(\mathcal{L}\varphi)(z) +O(\gamma),
$$
proving the main claim in this very simplified case. The general case, although much longer and tedious to check,  can still  be obtained in a similar straightforward manner.
\end{proof}

An {\em eigenmeasure} of $P$ is defined as a signed measure $\nu$ such that $\nu P=\lambda \nu$.
The density of an eigenmeasure relative to Lebesgue measure on the half-line will be referred to as
an {\em eigendensity} of $P$. They have the form $\phi \rho$, where $\rho$ is the stationary
probability density  and $\phi$ is an eigenfunction, $P\phi=\lambda\phi$. 
Based on
the proposition, the first few eigenvalues and eigendensities  of $P$ for sufficiently small values of $\gamma$
are expected to be approximated by those of the operator $I + 2\gamma^2 \mathcal{L}$. 
We do not study this spectral approximation
problem here (this is part of a more general study that will be presented elsewhere), but only point out the numerical agreement for the second eigenvalue and eigendensity shown
in Figure \ref{MBsecond}.

The spectral theory  for $\mathcal{L}$ corresponds to a standard Sturm-Liouville problem. 
In fact, under the change of coordinates $x=z^2/2$, $\mathcal{L}$ is  (up to a constant multiple $-2$)
the differential operator of Laguerre's equation
$$2  (x e^{-x} \phi_x)_x=\lambda e^{-x}\phi.$$
Polynomial
solutions   exist for $\lambda= -2n$, where $n$ is a non-negative
integer, and the corresponding eigenfunctions are easily obtained by textbook methods. The second polynomial
eigenfunction  (the first being the constant function) is given by $\phi(z)=1-z^2/2$, associated to eigenvalue $\lambda=-2$.
Thus  it is natural to expect that the spectral gap  of  $P$ is approximately $4\gamma^2$ for
small values of $\gamma$.

\vspace{0.1in}
\begin{figure}[htbp]
\begin{center}
\includegraphics[width=4.5in]{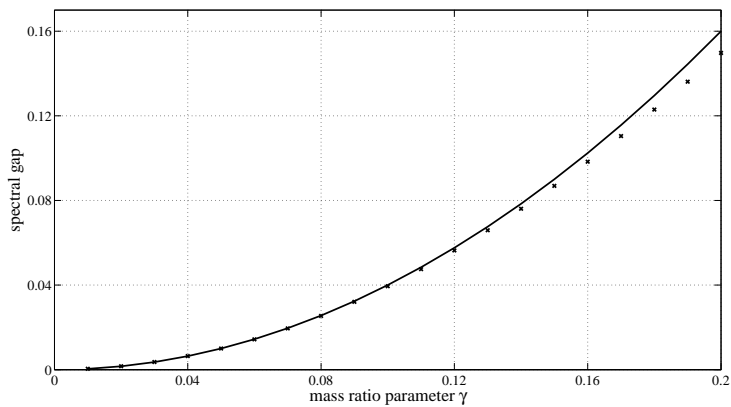}\ \ 
\caption{\small 
Asymptotics  of the
spectral gap of $P$ for small values  of the mass-ratio parameter $\gamma$. The discrete  
points are the values of the gap obtained numerically. The solid curve is the graph of $f(\gamma)=4\gamma^2$, suggested by
approximating $P$ by a second order differential equation.  }
\label{MBgap}
\end{center}
\end{figure}

\section{Other examples}\label{otherexamples}
We give a few further examples of simple systems to illustrate  the content of the main theorems.

\subsection{Wall systems without moving parts}
In this subsection we very briefly consider examples  having a trivial wall system, for which the Gibbs canonical distribution does
not make sense. These are nevertheless interesting, and we have studied them in some detail in previous papers.
(See \cite{F,FZ,FZsmallK}.) Our only concern here is to see how they   fit into the present more general 
set-up.

By assumption,  $M_{\text{\tiny wall}}$ reduces to a single point. If we further assume that  the molecule  is a point particle,
then the only dynamical variable of interest  is the velocity before and after the collision event, $v, V\in \mathbbm{H}$,
where $\mathbbm{H}$ is a half-space in $\mathbbm{R}^{k+1}$.
By conservation of energy,  $\|v\|=\|V\|$, so it suffices to take the hemisphere $S^+$ of unit vectors in $\mathbbm{H}$ as
the state space for the Markov chains.  
The only random variable is the   point in $\mathbbm{T}^k$, assumed to be uniformly distributed. Given   an essentially bounded function $f$  on $S^+$,
the Markov operator applied to $f$  takes the form
\begin{equation}\label{simpleoperator}(Pf)(v)=\int_{\mathbbm{T}^k} f(\Psi_v(x))\, dx, \end{equation}
where $\Psi_v(x)$ is the post-collision velocity with initial conditions $(x,v)$. A  first  example was suggested by
Figure \ref{exampleV}.  Planar billiards, as in Figure \ref{exampleI}, provide a large and interesting general class of examples of a purely
geometric nature, for which the operator $P$ is canonically determined by the billiard shape.

In the next proposition, let $J:S^+\rightarrow S^+$ be the linear involution that  sends the north pole to itself
and points on the equator to their antipodes.  We use the same symbol to denote the induced composition
operator on functions on $S^+$.  
A unit of the periodic contour  which, for the
example of Figure \ref{exampleI},  is   shown on the right-hand side of the figure,
will be called a  {\em billiard cell}. A billiard cell 
is 
{\em symmetric} if it is invariant  under $u\mapsto -u$ in $\mathbbm{T}^k$ (which induces the
map $J$ on velocities)

\begin{proposition}\label{propositionsimplesystems}
For the systems without internal moving parts as described by the operator $P$ in \ref{simpleoperator},
the probability measure $\mu$ on $S^+$ defined  by
$$ d\mu(v):=\frac{\Gamma(k/2 +1)}{\pi^{k/2}}\, \langle v, n\rangle \, dA(v), $$ is stationary, 
where $dA$ is the Euclidean $k$-dimensional volume  form on $S^+$, $n$ is the unit vector perpendicular to
the boundary of $\mathbbm{H}$, pointing towards the billiard surface, and the inner product is the
standard dot product. On the Hilbert space $L^2(S^+, \mu)$, the operator $P$ is bounded of norm $1$
and satisfies $P^*=JPJ$. If the billiard cell is symmetric, then $P$ is self-adjoint. 
\end{proposition}
\begin{proof}
This  follows from the general description of  invariant measures  of Theorem \ref{invariantvol}   and the fact that
the conditions
of Proposition \ref{selfadjoint} are satisfied. See also \cite{FZ}.
\end{proof}

 \vspace{0.1in}
\begin{figure}[htbp]
\begin{center}
\includegraphics[width=3.5in]{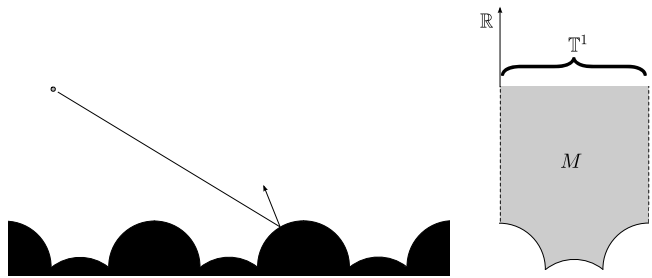}\ \ 
\caption{\small  Billiard collisions between a point particle and wall with periodic contour. 
A point particle reflects off  of a hard surface with periodic relief in ordinary billiard fashion. 
Here $\overline{M}_{\text{\tiny mol}}$ and $M_{\text{\tiny wall}}$ consist of   a single point each, and $M$ is 
the subset of $\mathbbm{R}\times \mathbbm{T}^1$ indicated on the right-hand side of the figure. It is given the induced  flat metric 
and the potential function is constant. }
\label{exampleI}
\end{center}
\end{figure}

The operator $P$    is often (and,  for  the specific contour  of Figure \ref{exampleI}, this is  a consequence of  results in \cite{F} or \cite{FZ}) a Hilbert-Schmidt operator.
A problem of particular interest is the   relationship between its spectrum of
eigenvalues and the geometric features of the billiard cell. We refer to 
\cite{FZ} and \cite{FZsmallK} for more  information.

Another   example to  which Proposition \ref{propositionsimplesystems} applies is shown in  Figure \ref{exampleII}.
In this case, the wall is featureless but the molecule is not:
  $M_{\text{\tiny wall}}$  consists of a single point and $\overline{M}_{\text{\tiny mol}}$ can be identified with the circle $S^1$.
We assume constant potentials. Since the wall is translation invariant, the  length scale for $\mathbbm{T}^1$ is not specified
(and not needed).
Let $l$ be the fixed length of the arm connecting the two  masses
 and  $m=m_1+m_2$  the total mass.  Let $(x,y)$ represent the coordinates of the center of mass of the dumbbell molecule. Then the configuration manifold is given by
$$M=\left\{([\theta], x, y): \min\left\{y-\frac{m_2}{m}l\sin\theta, y+\frac{m_1}{m}l\sin\theta\right\}\geq 0\right\}. $$
Here $[\theta]$ is  element in $S^1=\mathbbm{R}/2\pi\mathbbm{Z}$. A  $(\theta, y)$-cross section of $M$ is shown in
Figure \ref{exampleVI}.
By introducing the scaled angle coordinate
$$z:=\frac{\sqrt{m_1m_2}}{m} l\theta$$
  the kinetic energy, as a function of
the coordinates $(x,y,z,\dot{x}, \dot{y}, \dot{z})$ on the tangent bundle of $M$, takes the form
$$ E(x,y,z, \dot{x}, \dot{y}, \dot{z})=\frac12 m\left({\dot{x}}^2+ {\dot{y}}^2+{\dot{z}}^2\right),$$
which corresponds to the standard Euclidean metric in regions of $\mathbbm{R}^3$.
In terms of these new coordinates, collisions are described by ordinary specular reflection on the
boundary of $M$. (We are assuming here that the surface is perfectly smooth in the physical sense, i.e.,  there is no  
 tangential transfer of momentum 
between the particles and the surface.)

\begin{figure}[htbp]
\begin{center}
\includegraphics[width=3.5in]{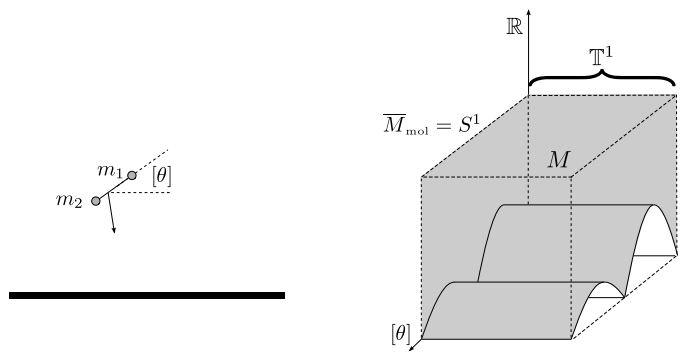}\ \ 
\caption{\small Collision of  a rotating dumbbell with a flat surface.  }
\label{exampleII}
\end{center}
\end{figure}

By restricting attention to a cross section ($x=$ constant)   this $3$-dimensional system can be reduced
to a $2$-dimensional system that is very much  like the one of Figure \ref{exampleI}, with
$z$   taking the role of  the length coordinate on $\mathbbm{T}^1$ in the  first example.  If we assume that  $z$ is random, uniformly distributed, 
then we have a system that is of essentially of the same kind as that of the previous example.

\begin{figure}[htbp]
\begin{center}
\includegraphics[width=2.5in]{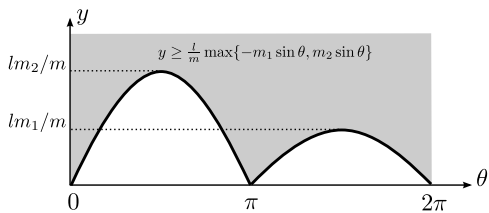}\ \ 
\caption{\small A cross section of the manifold with boundary depicted in Figure \ref{exampleII}.}
\label{exampleVI}
\end{center}
\end{figure}

Letting  the  normalized speed  $u=v/v_{\text{\tiny max}}$  of the center of mass of the dumbbell molecule be the variable of interest (assuming, for simplicity,
that the constant horizontal momentum is zero), then $P$
is regarded as a Markov operator with  state space $[0,1]$, where 
$v_{\text{\tiny max}}=\sqrt{2\mathcal{E}/m}$ is the maximal speed
that can be attained for a given, fixed, energy value. 
Writing  $u=\sin\theta$ for  $\theta\in [0,\pi]$,
 the  stationary distribution $\mu$  given by Proposition \ref{propositionsimplesystems}  has the form
 $$ d\mu(\theta)=\frac12 \sin\theta\, d\theta.$$
 It is easily shown that, for any initial probability distribution for $\theta$, the
 corresponding Markov chain is (Lebesgue measure)-irreducible and aperiodic, and $\mu$ is
 the unique stationary probability measure.

 Notice that    $P$ does not depend on the   length $l$ 
separating the two masses since changing $l$ only produces  a homothety  change of the rescaled 
region  of Figure \ref{exampleVI} (i.e., after the change of variables from $\theta$ to $z$). So $P$ only depends on the   
 mass-ratio. Let  $\gamma:=\sqrt{m_1/m_2}$.
 It is convenient to set $l=(\gamma^{-1}+\gamma)/2\pi$. With this choice, 
 the billiard cell contour in the $y,z$ coordinate  plane is bounded below by the graph of the
 function $$ y=\frac{1}{2\pi}\max\left\{-\gamma\sin(2\pi z), \gamma^{-1} \sin(2\pi z)\right\}$$
for  $z\in [0,1].$

Based on the results and arguments from \cite{F} and \cite{FZ}  one should expect    $P$ to
be an integral operator  (Hilbert-Schmidt, or at least quasi-compact). We do not attempt to show this here, but only offer the  numerical
observation about the dependence of the spectral gap of $P$ on the  mass-ratio parameter $\gamma$ shown in Figure
\ref{gapdumbbell}.
The graph exhibits   a great deal of structure which, at this moment, we do not know how to interpret.

  \vspace{0.1in}
 \begin{figure}[htbp]
\begin{center}
\includegraphics[width=4in]{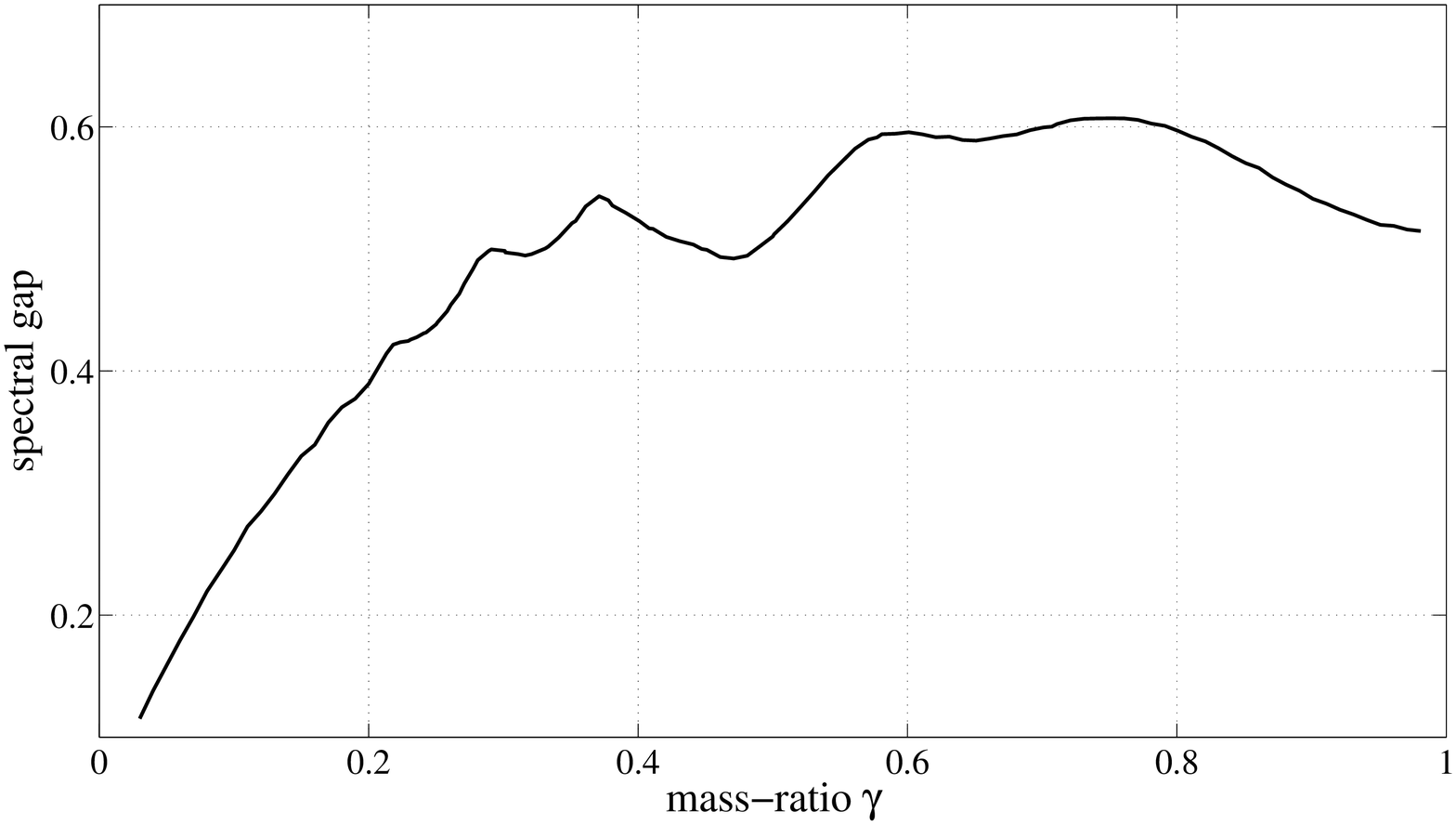}\ \ 
\caption{{\small   The spectral gap of $P$ as function of the parameter $\gamma=\sqrt{m_1/m_2}$}, where
$m_2>m_1$.}
\label{gapdumbbell}
\end{center}
\end{figure}

\subsection{Adding potentials}\label{examplepotential}

For an example with a  non-constant potential function, consider 
the system on the left-hand side of Figure \ref{exampleIII}. 
This is similar to the  two-masses system  of Section \ref{twomasses} except that we add a linear spring potential acting on mass $m_1$.
Thus we  consider   a spring-mass   with (essentially point) mass $m_1$, which comprises
the wall subsystem,  and a point mass $m_2$ corresponding to  the molecule subsystem. Suppose for simplicity that the
free mass $m_2$ can only move vertically,
so the whole set-up should be regarded as  one-dimensional. It is assumed that there are no  potentials involved other
than the elastic potential of the spring. (In particular, no gravity.)
 Then $M_{\text{\tiny wall}}$  is an interval, $\overline{M}_{\text{\tiny mol}}$ is
a single point, the torus component has  dimension $0$,  and  $M$ is the subset of $\mathbbm{R}\times M_{\text{\tiny wall}}$
indicated on the right-hand side of the figure.

The region of interaction is   the interval $[0,l]$, where $l$ is a positive number.  
Let $x_1$, $x_2$ indicate the positions of the masses 
$m_1$ and $m_2$
in physical space (on the left-hand side of the figure).
Using the scaled
coordinates 
$$x=\sqrt{\frac{m_1}{m}}\left(x_1-\frac{l}{2}\right), \ \ y= \sqrt{\frac{m_2}{m}}\, x_2$$
for the respective  positions of  $m_1$ and $m_2$,  the  energy function for the system with
  linear spring potential  is given by
$$E(x,y,\dot{x}, \dot{y})=\frac{m}2  \left({\dot{x}}^2 +{\dot{y}}^2+ \frac{k }{ m_1}  x^2\right).$$
 The   motion in $M$  between two  collisions (of the two masses or  of $m_1$ with the bottom wall or
 the semi-permeable wall at $l$) is  given  
by the functions of $t$:
\begin{equation}\label{equationsspringmass}
x(t)=x_0 \cos\left(\sqrt{\frac{k}{m_1}}\, t\right) + \dot{x}_0 \sqrt{\frac{m_1}{k}}\sin\left(\sqrt{\frac{k}{m_1}}\, t\right),\ \ 
y(t)=\dot{y}_0  t + y_0.
\end{equation}

 \vspace{0.1in}
 \begin{figure}[htbp]
\begin{center}
\includegraphics[width=3.0in]{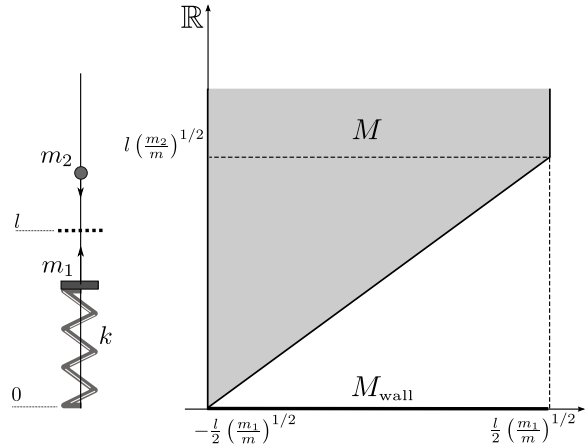}\ \ 
\caption{{\small  Collision of a point mass with a spring-mass system. }}
\label{exampleIII}
\end{center}
\end{figure}

The state variable of the Markov chain in this case is taken to be the speed $v=|\dot{y}|\in (0,\infty)$ of $m_2$.
It is assumed that the statistical  state of $m_1$ 
is a Gibbs distribution with parameter $\beta$. For concreteness, we describe chain transitions $v_{\text{\tiny old}}\mapsto v_{\text{\tiny new}}$
in algorithmic fashion:
\begin{enumerate}
\item Choose independent, uniform  random numbers $U_1$, $U_2$ in $[0,1]$, and a sign  $s\in \{-,+\}$ with equal probabilities;
\item Let  $\mathcal{E}=-\frac1\beta \ln U_1$; thus $\mathcal{E}\in (0,\infty)$ has probability density $\beta \exp(-\beta \mathcal{E})$;
\item Let $L(\mathcal{E}):=\min\left\{1, \frac{l}{2}\left(\frac{k}{2\mathcal{E}}\right)^{1/2}\right\}$ and
$x=\left(\frac{2\mathcal{E}m_1}{m k}\right)^{1/2} \sin\left((2U_2-1) \arcsin(L(\mathcal{E}))\right);$ thus $x$
has probability density proportional to $h_{\mathcal{E}}^{-1}$
over the interval
$  -L(\mathcal{E}) \sqrt{{2\mathcal{E}m_1}/{m k}}  \leq  x \leq L(\mathcal{E}) \sqrt{{2\mathcal{E}m_1}/{m k}},$
where  
$ h_{\mathcal{E}}(x):=\sqrt{2(\mathcal{E}-U(x))}$ and $U(x)=kmx^2/m_1$ is the spring potential in the scaled coordinate;
\item Set $(x, s h_{\mathcal{E}}(x))$ to be the state of the mass $m_1$ when $m_2$ enters the region of interaction,
and $(l\sqrt{m_2/m}, v_{\text{\tiny old}})$ the state of $m_2$. Let the system evolve, deterministically with this initial condition until
$m_2$ is back at position $l\sqrt{m_2/m}$. Along the way, assume that collisions with the boundary  of  the two dimensional region
on the right side of Figure \ref{exampleIII} are specular and in between collisions the trajectory satisfies 
equations \ref{equationsspringmass}. When $m_2$ reemerges at $l\sqrt{m_2/m}$ set $v_{\text{\tiny new}}$ equal to its
new speed.
\end{enumerate}

By Theorems \ref{gibbs} and   \ref{invariantvol}, this Markov process has stationary probability  given by the Maxwell-Boltzmann
distribution  (after reverting to
the variables prior to scaling, with speed $u=\sqrt{m/m_2}v$)
\begin{equation}\label{stationaryMB} d\mu(u)=\beta m_2 u \exp\left({-\beta\frac{m_2 u^2}{2}}\right)\, du.\end{equation}

\begin{figure}[htbp]
\begin{center}
\includegraphics[width=3.5in]{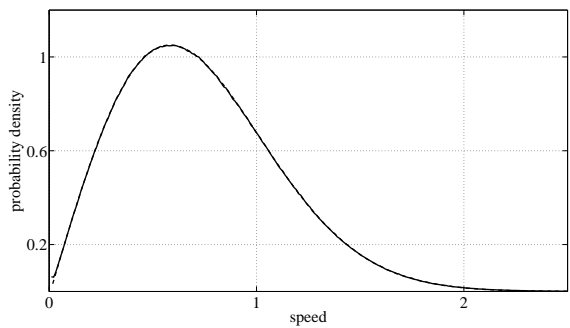}\ \ 
\caption{{\small The figure superposes the graph of the stationary distribution \ref{stationaryMB} and 
the same distribution obtained by numerically simulating the Markov chain according to the algorithm 
described above. }}
\label{MBspeedCom}
\end{center}
\end{figure} 

In comparing theses distributions with the corresponding textbook expressions, the reader should keep in mind the distinction between
the Maxwell-Boltzmann distribution in the interior of  billiard table (the gas container) and the similar distribution on the 
boundary surface (wall). The latter  has density proportional to 
$$ \beta m \,  v\cdot n  \|v\|^{d-1} \exp\left(-\beta \frac{m\|v\|^2}{2}\right)$$
in dimension $d$, where $n$ is the unit normal to the wall surface.

\section{Invariant volume forms}\label{invariantvolumes}
The main purpose of this section is to prove Theorems \ref{gibbs} and   \ref{invariantvol}.

\subsection{Definitions}
Recall the function $d:M\rightarrow \mathbbm{R}$  introduced in Subsection \ref{setupintro}. 
The set   $S:=\overline{M}_{\text{\tiny mol}}\times \{0\} \times \mathbbm{T}^k \times M_{\text{\tiny wall}}$
is the level set $d=0$. It will be convenient  in this section to  disregard  the part of $M$ given by $d>0$ and consider  $S$ as 
 a submanifold of the boundary of $M$.
Observe that $S$ lies in the interior of the   product region,  where the molecule and wall  subsystems are
    non-interacting.   
  Thus  it makes sense to define over
    a neighborhood of $S$ in $M$  the unit vector field $\nu$ along the $\mathbbm{R}$-factor of $M_{\text{\tiny mol}}$.
    We choose  the direction of $\nu$ so that  it points
towards the region of interaction. 
    The restriction of $\nu$ to $S$ is then  a unit vector field perpendicular to $TS$, pointing into $M$.
Let   $N:=TM$ and   $N_S:=\iota^*TM$, the pull-back of $TM$ to $S$ under the inclusion map
$\iota:S\rightarrow M$. Also define the subset
 $N_S^+$ of $N_S$ consisting of $v$ such that $\langle v, \nu\rangle>0$,
 $N_S^-:=-N_S^+$, and $N_S^0:= TS$. These are all bundles over $S$. We 
 often denote fibers of a bundle using subscripts, as in $N_q=T_qM$. When this is inconvenient, we use
 function form, so that
$N_S(q)$, for example, is the fiber of $N_S$ above $q\in S$.  Projection maps 
for  these bundles  will   be denoted by the same symbol  $\tau$. Projection maps for other fibrations will
typically be denoted by $\pi$.

Clearly, the reflection map $R$ maps $N_S^\pm$ to $N_S^\mp$ and $N_S^0$ to itself. 
More generally, we can define $N_\partial$   as the restriction (pull-back) of $N$ to $\partial M$,
and $N_\partial^\pm$  as we did in the case of $N_S$. The notation $N^{\partial\pm}$ may  also be  used when convenient. On
regular points of the boundary, the reflection map is defined on $N_\partial$. 
If we wish to emphasize that something is taking place over regular points of the boundary (or regular points of the function
$E$), we may indicate  this by adding  a subscript such as in $N^{\partial+}_{\text{\tiny reg}}$.

 Let $\langle\cdot, \cdot \rangle$ be a Riemannian metric on $M$ and $U:M\rightarrow \mathbbm{R}$ the potential function.
 Unless explicitly stated otherwise, functions and tensor fields are assumed to be smooth on   interior  and
regular boundary points. The Riemannian metric defines the {\em kinetic energy} function $\kappa:TM\rightarrow\mathbbm{R}$
 given  by $$\kappa(q,v):= \frac12 \|v\|_q^2:=\frac12 \langle v, v\rangle_q. $$  
 The (total) {\em energy} function of the Newtonian system on $M$  with potential  $U$ is
 $E:TM\rightarrow \mathbbm{R}$ such that $$E(q,v):=\kappa(q,v)+U(q).$$
 We write   $N(\mathcal{E}):=\{(q,v)\in N:E(q,v)=\mathcal{E}\}$.
 Clearly, the base point projection   $\tau :N\rightarrow M$   
 maps 
  $N(\mathcal{E})\setminus \text{zero section}$ into $ M^\mathcal{E}:=\{q\in M: U(q)<\mathcal{E}\}.$
 The intersection of $N(\mathcal{E})$ with $N_S$ is denoted $N_S(\mathcal{E})$.  Similar notations are
 used for the various related sets defined earlier. 
Whenever convenient, we specify points in  $N$   simply by $v$ instead of $(q,v)$. 
 For example, we typically write $q=\tau(v)$,   $T_vN$, $d\tau_v$, etc.,  for $v\in N$.

Let  $\nabla$ be the Levi-Civita covariant derivative operator. If $X(s)$ is a vector
field along  a differentiable curve $s\mapsto c(s)$ such that $v=c'(0)$,  the covariant derivative of $X$ along
$c$ at $s=0$ will be written   
   $\left.\frac{\nabla X}{ds}\right|_{s=0}$ or, when appropriate,   $\nabla_vX$.   
The  {\em horizontal bundle} $H$ over $N$ is the subbundle of $TN$ defined  as  the kernel of 
 the {\em connection map} $K_v:T_vN\rightarrow T_{\tau(v)}M$, which is derived from $\nabla$ as follows.
  Let $v'(0)\in T_vN$,
where $s\mapsto v(s)$ is a differentiable  curve in $N$ such that $v(0)=v$; then
$$K_vv'(0)=\left.\frac{\nabla}{dt}\right|_{t=0}v(t) .$$
Write $\xi:=v'(0)$ and  $w:=d\pi_v \xi$. Let $Y$ be a
smooth section of $\pi$ over a neighborhood of $q=\pi(v)$ such that $Y(q)=v$
and $dY_q w=\xi$. Then
$K_v\xi =\nabla_wY$.

Let $V$ denote the {\em vertical bundle}, which is the vector bundle  over $N$ whose
fiber above a $v\in N$ is the tangent space to $N_{\pi(v)}$ at $v$.    
Thus $V_v$ is the kernel of $d\pi_v$, 
the projection $d\pi_v:H_v \rightarrow T_qM$ is a  linear isomorphism, 
and the direct sum decomposition  $TN=H\oplus V$ holds.
If $X$ is a smooth vector field on $M$, then the horizontal lift of $X$ is the smooth section of $H$ given by
$$v\mapsto X^H(v):=(d\tau_v|_{H_v})^{-1}X(\tau(v)).$$

For each $v\in N$ with $q=\tau(v)$, define
the linear isomorphism 
$ l_v:T_qM\rightarrow V_v$, called the {\em vertical lift} map, by 
$$l_v(w):= \left.\frac{d}{ds}\right|_{s=0}\left(v+sw\right). $$
The alternative notation $w^V(v)$ or $w^V_v$ will also be used later  instead of  $l_v(w)$. 
If $X$ is a smooth vector field on $M$, then the vertical lift of $X$ is the
smooth section of $V$ given by   $$v\mapsto X^V(v):=(X(\tau(v)))^V_v.$$

\subsection{Contact, symplectic, and volume forms}\label{invariantforms}
The manifold $N$ is equipped with the canonical contact form $\theta$  defined by
$\theta_v(\xi):=\langle v, d\tau_v \xi\rangle_q$, for $v\in N(q)$. It is well-known that $d\theta$
is non-degenerate, hence a symplectic form  on $N$. In terms of the
Riemannian metric, 
\begin{equation}\label{sympmetric} d\theta_v(\xi_1, \xi_2)=\langle K_v\xi_1, d\tau_v\xi_2\rangle -\langle K_v \xi_2, d\tau_v \xi_1\rangle\end{equation}
from which it easily follows that $d\theta$ is indeed non-degenerate and that $V$ and $H$ are Lagrangian subbundles.

 The {\em Hamiltonian vector field} (associated to the energy function $E$) is the vector field $X^E$ on $N$ 
 such that 
\begin{equation}\label{defX} X^E\righthalfcup d\theta = -dE.\end{equation}
One easily shows that  $E$ and $d\theta$ are invariant under $X^E$. Thus $X^EE=0$    and $\mathcal{L}_{X^E} d\theta=0$, where $\mathcal{L}_X$ indicates the Lie derivative along $X$.
 The contact form $\theta$, however, is not in general invariant   but satisfies
 $\mathcal{L}_X\theta=dL$, where $L$ is the Lagrangian function $L(q,v):=\kappa(q,v) -U(q)$. 
The Hamiltonian vector field can be written   as
\begin{equation}\label{HamVF} X^E=Z - (\text{grad } U)^V\end{equation}
where the {\em geodesic spray} $Z$  is the vector field   on $N$ defined, at each $v$, as the horizontal
 lift of $v$ to $H_v$.  In particular, $d\pi_v Z(v) = v$ and 
 $d\pi_v X^E=v$. 
Observe that $(q,v)$ is a critical point of $E$ exactly when $X^E(v)=0$, which can only
happen when $v=0$ and $q$ is critical for $U$.

The {\em Hamiltonian flow} is the (local) flow of $X^E$, which we denote by $t\mapsto \Phi_t:=\Phi_t^{X^E}$.
The flow lines project under $\pi$ to curves $c(t)$ on $M$ that satisfy Newton's equation, and
any solution of Newton's equation lifts to  a flow line in $N$. 
It will be convenient to let $t\mapsto \Phi_t(v)$ represent a trajectory of the system through its entire history,
which may include collisions and reflections with the boundary of $M$.  The Hamiltonian vector field is essentially
complete, in the sense defined earlier in Section \ref{setupintro} (part (iv) of list of assumptions).

\begin{proposition}\label{tangent} Let  $v\in N(\mathcal{E})$ be  a regular point for  $E$  and let $q=\tau(v)$ be either an
interior point of  $M$ 
or  in the regular boundary  $\partial_{\text{\tiny r}} M$. Then   
  $T_vN(\mathcal{E})$  consists of $\xi\in T_vN$ such that 
\begin{equation}\label{propE} d(U\circ \tau)_v \xi + \langle v, K_v \xi \rangle_{q} =0.\end{equation}
 The    projection $d\tau_v:T_vN(\mathcal{E})\rightarrow  T_{q}M$  is surjective. In fact,  for each $w\in T_{q}M$,
$$w^{\curlywedge}(v) := w^H(v) - \|v\|^{-2} {dU_q(w)}W(v)\in T_vN(\mathcal{E})$$
satisfies $d\tau_v w^{\curlywedge}(v)=w$.
The space   $T_v N^{\partial+}(\mathcal{E})$ consists of all   $\xi\in T_vN$ such that 
    \ref{propE} 
holds and  $d\tau_v \xi$ is   tangent to $\partial M$.
\end{proposition}
\begin{proof} These claims are easily derived by observing that 
 Equation \ref{sympmetric},  the definition of $X^E$ in \ref{defX}, and Equation \ref{HamVF}, imply
$$
dE_v(\xi)  = d\theta_v(\xi, X^E)\\
= \left\langle K_v \xi, d\tau_v X^E\right\rangle - \left\langle K_v X^E, d\tau_v \xi\right\rangle\\
=\langle v, K_v\xi\rangle + \langle \text{grad } U, d\tau_v \xi \rangle.
$$
Of course,   $\xi$ lies in $T_vN(\mathcal{E})$ iff $dE_v(\xi)=0$ on regular points of $E$.
\end{proof}

\begin{proposition}\label{contactR}
Let $N^\partial_{\text{\tiny reg}}(\mathcal{E})$ be the subset of $N^{\partial}(\mathcal{E})$
  of $v$ such that $\tau(v)$ is a regular boundary point of $M$ and
 $\iota:N^\partial_{\text{\tiny reg}}(\mathcal{E})\rightarrow N$    the inclusion map.
Then $\iota^* d\theta$ is a symplectic form on $N^\partial_{\text{\tiny reg}}(\mathcal{E})\setminus T(\partial_{\text{\tiny reg}}M)$. Furthermore,
the reflection map $R:N^\partial_{\text{\tiny reg}}(\mathcal{E})\rightarrow N^\partial_{\text{\tiny reg}}(\mathcal{E})$   leaves $\iota^*\theta$ and $\iota^* d\theta$ invariant.
\end{proposition}
\begin{proof}
Here and in what follows, we assume that all  points under consideration are regular for $E$ and, if  on
the boundary, that they are regular boundary points.   If  $v\neq 0$, then $v$ is necessarily regular
for $E$. With this    in mind, we omit references to `regular' in
the notation from now on.

The main issue is to check that   $\omega:=\iota^* d\theta$ is non-degenerate. Let $\xi \in T_vN^\partial(\mathcal{E})$ and 
assume that $\omega_v(\xi, \eta)=0$ for all $\eta\in T_v N^\partial(\mathcal{E})$, where  $v$ is not
tangent to the boundary of $M$. First choose $\eta \in V_v$ such that $K_v \eta$ is orthogonal  to $v$. Then
$0=\omega_v(\xi,\eta)=-\langle K_v\eta, d\tau_v\xi\rangle $, so $d\tau_v\xi$ is orthogonal  to all $w$ such that $w\perp v$.
That is, $d\tau_v\xi$ is a scalar multiple of $v$.  By Proposition \ref{tangent} this vector is tangent to the boundary, thus
 zero by assumption, and $\xi$ must be a vertical vector. Let now $\eta$ be an arbitrary tangent vector to $N^\partial(\mathcal{E})$ at $v$.
 Then as 
 $0=\omega_v(\xi,\eta)=\langle K_v\xi, d\tau_v\eta\rangle $  and $d\tau_v \eta$ is tangent to $\partial M$ we conclude that 
 $K_v\xi$ is a scalar multiple of the normal vector $n$ at $\tau(v)$, that is, $\xi=\lambda n^V(v)$ for some $\lambda\in \mathbbm{R}$.
 From Proposition \ref{tangent} it follows that $0=d(U\circ \tau)_v\xi + \langle v, K_v\xi\rangle_{\pi(v)}=\lambda \langle v,n\rangle_{\tau(v)}$,
 which implies that $\lambda=0$. Therefore, $\xi=0.$
\end{proof}

It is useful to introduce the {\em Sasaki metric} on $N$, which is the Riemannian metric defined by
$$ \langle \xi, \eta\rangle_v :=\langle d\tau_v \xi, d\tau_v \eta \rangle_q + \langle K_v\xi, K_v \eta\rangle_q$$
for all $\xi, \eta\in T_vN$ and $q=\tau(v)$.  In terms of  this metric  the vertical and horizontal subbundles are
mutually orthogonal and $H_v$, $V_v$ are isometric to $T_qM$ under $d\tau_v$ and $K_v$, respectively.

Define the vector field 
 $$\eta:=(\text{grad } E)/\|\text{grad } E\|^2,$$  the gradient and norm being associated to  the Sasaki metric.
 Observe that  $dE(\eta)=1$, so that 
 $ E\circ \Phi^\eta_s = E+s,$
 where $\Phi^\eta_s$ denotes the local flow of $\eta$. 
 It is not difficulty to obtain the expressions
 $$ \text{grad }E= W +(\text{grad U})^H, \ \ \|\text{grad }E\|_v^2=\|v\|^2_q + \|\text{grad }U\|_q^2$$
where $W$ is the {\em canonical vertical vector field},  defined by $W_v:=v^V_v$ for each $v\in N$.

Let   $\Omega:= (d\theta)^m$, where $m=\dim M$. Then $\Omega$ is a
volume form (i.e., a non-vanishing form of top degree) on $N$. 
It is also  invariant under the Hamiltonian flow  since $\mathcal{L}_Xd\theta=0$.

 \begin{proposition}\label{XintoOmega} Let $X^E$ be the Hamiltonian vector
 field for  the energy function $E$, and  $\eta$ the vector field introduced in the previous paragraph.
  Define $\Omega^E:=\eta\righthalfcup \Omega$.  Then 
 \begin{enumerate}
 \item $\Omega= dE\wedge \Omega^E$;
 \item $X^E\righthalfcup \Omega^E =  m (d\theta)^{m-1} + m(m-1) (\eta\wedge d\theta) \wedge  dE \wedge (d\theta)^{m-1}$;
  \item  
 $\mathcal{L}_{X^E}\Omega^E =dE\wedge(\eta\righthalfcup \mathcal{L}_{X^E}\Omega^E).$ 
 \end{enumerate}
 It follows  that
 the restriction of $\Omega^E$ to  each level set  $N(\mathcal{E})$ is non-vanishing (a volume form) and
 invariant under the Hamiltonian flow, and that the restriction of $X^E\righthalfcup \Omega^E$ to the same level
 sets equals $m (d\theta)^{m-1}$.
 \end{proposition}
 \begin{proof}
 For part 1, write $dE\wedge \Omega^E=f\Omega$ and take the interior multiplication on both side with $\eta$
 to conclude that $f\Omega^E=\Omega^E$. As $\Omega^E$ does not vanish, we have $f=1$. 
 For part 2, observe that
$$
X^E\righthalfcup (\eta\righthalfcup \Omega)
 = X^E\righthalfcup       \left(m\ \! (\eta\righthalfcup d\theta)\wedge (d\theta)^{m-1}\right)
 =m\ \! d\theta\left(\eta,X^E\right)(d\theta)^{m-1}-m(\eta\righthalfcup d\theta) \wedge \left(X^E\righthalfcup (d\theta)^{m-1}\right).$$
 But $d\theta(\eta,X^E)=dE(\eta)=1$  by the definition of $X^E$ and $\eta$. Also  $X^E\righthalfcup (d\theta)^{m-1}=-(m-1)dE\wedge (d\theta)^{m-2}$.
 For part 3, obtain  from $\mathcal{L}_{X^E} (dE\wedge \Omega)=0$ that $dE\wedge \mathcal{L}_{X^E}\Omega^E=0$
 and take the interior multiplication on both sides of this equation with $\eta$.   
 \end{proof}

\subsection{Billiard maps}
We fix an energy value $\mathcal{E}$ and assume
that $N(\mathcal{E})$ has finite volume relative to $\Omega^E$. 
Here we will  let  $S$   denote more generally than before (the regular part of) a submanifold of the boundary of $M$.
For any given $v\in N_S^+(\mathcal{E})$, define
  $$\mathcal{T}(v):=\inf \{t>0: \Phi_t(v)\in N^+_S(\mathcal{E})\},$$
which is $\infty$ if the flow line never returns to $ N^+_S(\mathcal{E})$.
By Poincar\'e's recurrence   applied to the Hamiltonian flow, $\mathcal{T}$ is finite with probability $1$ with respect to the
flow-invariant probability measure derived from $\Omega^E$.  Now define
the return map $N_S^+(\mathcal{E})$ by  $T:=R\circ \Phi,  \text{ where } \Phi(v):=\Phi_{\mathcal{T}(v)}(v),$
$R$ being  the reflection map.  Then $T$ is almost everywhere defined, and 
by one of our standing assumptions   it   is almost surely smooth.
(Section \ref{setupintro}, assumption (v); see  \cite{CM} for how this point concerning smoothness   is argued in the   simpler case of plane  billiards.)
We  denote by $\Omega^{E,S}$ the pull-back of $X^E\righthalfcup \Omega^E$ to $N_S^+(\mathcal{E})$
under the inclusion map. By  Proposition \ref{XintoOmega} this form agrees with the pull-back of $m(d\theta)^{m-1}$.

\vspace{.1in}
\begin{proposition}\label{return}
The  return map $T: N_S^+(\mathcal{E})\rightarrow N_S^+(\mathcal{E})$ preserves 
$\Omega^{E,S}$ almost everywhere.
\end{proposition}
\begin{proof}
This involves a standard argument, which we briefly recall. Let $v\in N_S^+(\mathcal{E})$  admit  
a neighborhood $\mathcal{U}$  where  
 $T$ is smooth. Let $c:[0, \mathcal{T}(v)]\rightarrow N(\mathcal{E})$ be the orbit segment connecting $v$ to $\Phi_{\mathcal{T}(v)}(v)$,
 and  $\gamma_1$ a  closed curve contained in $\mathcal{U}$. Let $D$ be a smooth embedded  disc contained in $\mathcal{U}$
 that  is bounded by $\gamma_1$, and denote by $\gamma_2$ the image of $\gamma_1$ under   $\Phi$. Then $\gamma_1$ sweeps out
 a surface $\Sigma$ under the Hamiltonian flow such that the boundary of  $\Sigma$ is the union of  $\gamma_1$ and $-\gamma_2$,
 where the negative sign indicates orientation.
 Notice that the restriction of $d\theta$ to $\Sigma$ is $0$ as $E$ is constant on this surface   and
 the interior multiplication of $d\theta$ by the Hamiltonian vector field is
 $-dE$. So $$0=\int_\Sigma d\theta=\int_{\gamma_1}\theta - \int_{\gamma_2} \theta=\int_{\gamma_1}\left[\theta -\Phi^*\theta\right]=\int_D d\left[\theta-\Phi^*\theta\right].$$
 As $\gamma_1$ and $D$  can be made arbitrarily small, we conclude that $\Phi^*d\theta=d\theta$. Since $R$ also preserves $d\theta$
 according to Proposition \ref{contactR},
the same holds for $T$. Therefore, $T$ leaves $\Omega^{E,S}$ invariant as claimed.
\end{proof}

\subsection{Product systems}\label{productsystems}
 
 We   next  specialize some of the above facts to   product systems. The notation here is independent 
 of that of the rest of the paper.  Let  $M=M_1\times M_2$. 
  Let $\tau_i:N_i:=TM_i\rightarrow M_i$ and
 $\tau:N:=TM\rightarrow M$ 
 be the   tangent bundle maps
 and let   $\pi_i$   be the projection $M\rightarrow M_i$.  The induced projection $N\rightarrow N_i$
 will also  be written   $\pi_i$, so
    it makes sense to write $\pi_i \circ \tau = \tau_i \circ \pi_i$.  If there is some  possibility of confusion
    we may write, for example, $(q_i,v_i)=(q_i, (d\pi_i)_q v)$ instead of $v_i=\pi_i(v)$ for a given $v$ in $ N$. 
    Either way, the
    product Riemannian metric reads
       $$\langle v, w\rangle_q =\langle v_1,w_1\rangle_{q_1}+\langle v_2,w_2\rangle_{q_2}. $$
       Vertical and horizontal lifts, and the corresponding subbundles of $TN$ decompose as expected in terms
       of the respective notions on $N_i$. In particular,
       the Sasaki metric is similarly decomposed as $\langle\cdot, \cdot \rangle=\pi_1^*\langle\cdot,\cdot \rangle_1 + \pi_2^*\langle\cdot,\cdot \rangle_2$.
The canonical contact form $\theta$ on $N=N_1\times N_2$ becomes
$\theta=\pi_1^*\theta_1 + \pi_2^*\theta_2$, where $\theta_i$ is the contact form on $N_i$, and
the invariant volume form $\Omega=\pi_1^*\Omega_1 \wedge \pi_2^*\Omega_2$.
Whenever  convenient, we omit explicit reference to the projection maps and write, for example, $\Omega=\Omega_1\wedge\Omega_2$
or $\theta=\theta_1+\theta_2$. 

Assuming  that the potential function $U$  on $M$ has the form
$U=U_1\circ \pi_1 + U_2\circ \pi_2$, where $U_i$ is a smooth function on $M_i$, the energy function becomes
$E=E_1\circ\pi_1 +E_2\circ\pi_2$ and the Hamiltonian vector field on $N$ is written as $X^E=X_1+X_2$, where $X_i$ is
characterized by being
$\pi_i$-related to the Hamiltonian vector field on $N_i$ associated to $E_i$ and $\pi_j$-related to $0$ for $j\neq i$.
The (Sasaki) gradient of $E$  will be written, with slight abuse of notation, as
$\text{grad }E=\text{grad }E_1 + \text{grad }E_1$ and the vector field $\eta:=(\text{grad }E)/\|\text{grad }E\|^2$
becomes
$$ \eta=\frac{\|\text{grad }E_1\|^2}{\|\text{grad }E_1\|^2+\|\text{grad }E_2\|^2}\eta_1 +\frac{\|\text{grad }E_2\|^2}{\|\text{grad }E_1\|^2+\|\text{grad }E_2\|^2}\eta_2. $$
\begin{proposition}\label{productinvariant}
Let $\Omega^E:=\eta\righthalfcup \Omega$ be the invariant volume form on the energy level $N(\mathcal{E})$
and similarly define $\Omega^{E_i}_i:=\eta_i\righthalfcup \Omega_i$ on level sets $N_i(\mathcal{E}_i)$.
Then the level sets $N(\mathcal{E})$ can be measurably partitioned    as a disjoint union of
product manifolds $$N(\mathcal{E})=  \bigsqcup N_1(\mathcal{E}_1)\times N_2(\mathcal{E}-\mathcal{E}_1)$$ 
where the elements of the partition are the level sets of $E_1: N(\mathcal{E})\rightarrow \mathbbm{R}$, and the
invariant volume $\Omega^E$ has the decomposition
$$ \Omega^{E}=dE_1\wedge\Omega_1^{E_1}\wedge \Omega_2^{E_2}$$
adapted to this partition.
\end{proposition}
\begin{proof}
The main point is to verify  the stated  form of $\Omega^{E}$.
Define   $\alpha_i$  by    $\eta=\alpha_1 \eta_1 +\alpha_2 \eta_2$. 
Now,  $\Omega^{E}$ can be written as
$$
(\alpha_1\eta_1 +\alpha_2 \eta_2)\righthalfcup (\Omega_1 \wedge \Omega_2)=
\alpha_1\Omega_1^{E_1}\wedge \Omega_2 +\alpha_2 \Omega_1\wedge\Omega_2^{E_2}= \alpha_1 \Omega_1^{E_1}\wedge dE_2\wedge \Omega_2^{E_2}+\alpha_2 dE_1\wedge  \Omega_1^{E_1}\wedge \Omega_2^{E_2}.
$$
Since   $dE_2=-dE_1$ on $N(\mathcal{E})$, $\Omega_1^{E_1}$ is an odd-degree form, and $\alpha_1+\alpha_2=1$,
$$ \Omega^{E}=(-\alpha_1 dE_2+\alpha_2 dE_1)\wedge \Omega_1^{E_1}\wedge\Omega_2^{E_2}=dE_1\wedge  \Omega_1^{E_1}\wedge\Omega_2^{E_2}$$
as claimed.
\end{proof}

The Gibbs canonical distribution on $N_i$ with temperature parameter $\beta_i$ is the probability measure
on $N_i$ defined by   the form
$$ \zeta_i:= \frac{e^{-\beta_i E_i}}{Z_i(\beta_i)} \Omega^{E_i}_i \wedge dE_i,$$
where $Z_i(\beta_i)$ is a normalization constant. The following trivial but key observation must be noted.

\begin{corollary}\label{keycorollary}
If the states of the  two  subsystems  are distributed according to  the Gibbs canonical distribution with same parameter $\beta$, then the state of the product system is also distributed according to the Gibbs canonical distribution with parameter $\beta$.
\end{corollary}
\begin{proof}
Let $E=E_1+E_2$ and define $Z(\beta)=Z_1(\beta)Z_2(\beta)$. Note that
$dE_1\wedge dE_2=dE_1\wedge dE$. Due to Proposition \ref{productinvariant}, 
$$\zeta_1\wedge \zeta_2=\pm  \frac{e^{-\left(E_1+E_2\right)}}{Z_1(\beta)Z_2(\beta)} \Omega^{E_1}_1\wedge \Omega^{E_2}_2 \wedge dE_1\wedge dE_2=\pm \frac{e^{-\beta E}}{Z(\beta)}\Omega^E\wedge dE.$$
The measure obtained  from $\zeta_1\wedge \zeta_2$ is already normalized,   so $Z(\beta)$ is
the correct denominator.  
\end{proof}

Theorem \ref{gibbs} can now  be seen to follow from Corollary \ref{keycorollary} and Proposition \ref{invariantstationary}.
If the state of the  wall system has the
Gibbs distribution with parameter $\beta$ and the state of the   molecule system is given, prior to entering the interaction zone,  the Gibbs distribution with the 
same parameter, then the state of the  joint (product) system has a probability distribution which is invariant under the deterministic  return map to
the non-interaction zone. Thus the
molecule factor of the  state distribution of the total system upon return to the non-interaction zone remains  the same.

\subsection{Frame description of the volume forms}\label{framevolumes}
Let $m$ be the dimension of $M$ and $U\subset M$  an open subset on which is defined a smooth  orthonormal 
frame of vector fields $\{e_1,\dots, e_m\}$.  
Let $N_U$ be the subset of elements in  $N$ with base point in $U$. Define 
$$\mathcal{N}_U:=\{(q, u, \mathcal{E})\in U\times \mathbbm{R}^m\times \mathbbm{R}:\|u\|=1 \text{ and } U(q)<\mathcal{E}\}.$$
Thus $\mathcal{N}_U$ is an open submanifold of $U\times S^{m-1} \times \mathbbm{R}$. 
Let $\mathcal{N}_U(\mathcal{E})$ denote the submanifold mapping to $\mathcal{E}$ under $\pi_3:\mathcal{N}_U\rightarrow \mathbbm{R}$.
Let $\{c_1, \dots, c_m\}$ represent the standard basis of $\mathbbm{R}^m$ and $u\cdot v$ the ordinary inner product. 
Observe that 
$$ c_i^\curlyvee(u):=\left.\frac{d}{ds}\right|_{s=0}\frac{u+sc_i}{\|u+sc_i\|}=c_i-u\cdot c_i\,  u$$
is tangent to $S^{m-1}$ at a unit vector  $u$, and    $\{c^\curlyvee_1, \dots, c^\curlyvee_{m-1}\}$ is a basis  of $T_uS^{m-1}$ for all $u$  not
perpendicular
to $c_m$. For these $u$, let $\{\varphi_1, \dots, \varphi_{m-1}\}$ be the dual basis associated to  $\{c^\curlyvee_1, \dots, c^\curlyvee_{m-1}\}$.
In terms of this dual basis, the (standard) Riemannian volume form on $S^{m-1}$ is
\begin{equation}\label{volumesphere}\omega_u^{\text{\tiny sphere}}= (-1)^{m-1}u\cdot c_m\, \varphi_1^\curlyvee \wedge \dots \varphi_{m-1}^\curlyvee.\end{equation}
This is obtained  from   $\omega_u^{\text{\tiny sphere}}=u\righthalfcup (c_1^*\wedge \dots \wedge c_m^*)$, where the $c_i^*$ 
constitute  the dual standard basis, by  evaluating this form on the vectors $c_i^\curlyvee(u)$.
The Riemannian volume form on $M$ (up to sign) is 
$$\omega^M:=e^*_1\wedge \dots \wedge e_m^*, $$
where the $e_i^*$ form the dual frame on $T^*M$. By identifying $e_i$ with $(e_i, 0, 0)$ and $c_i^\curlyvee$ with
$(0, c_i^\curlyvee, 0)$, we may think of $e_i$ and $c_i^\curlyvee$ as tangent to $\mathcal{N}_U(\mathcal{E})$,
and $\{e_1^*, \dots, e_m^*, \varphi_1^\curlyvee, \dots, \varphi_{m-1}^\curlyvee\}$ as a  frame of $1$-forms on  $\mathcal{N}_U(\mathcal{E})$.

We now introduce  a diffeomorphism $F:\mathcal{N}_U\rightarrow N_U$ by 
\begin{equation}\label{diffeomorphism} F(q,u,\mathcal{E})=\left(q, h_\mathcal{E}(q)\sum_i u_i e_i(q)\right),\end{equation}
where  $h_{\mathcal{E}}(q):=\sqrt{2(\mathcal{E}-U(q))}$.
The inverse map is 
$F^{-1}(q, v)=(q, u, \mathcal{E})$,  where 
$u_i=\|v\|^{-1}_q\langle v, e_i(q)\rangle_q$  and  $\mathcal{E}=\frac12\|v\|^2_q +U(q)$.

\begin{proposition}\label{curly}
For any given    $v, w\in N_q$, define vectors $w_v^\curlyvee$ and $w_v^\curlywedge$ in $T_vN$ by
$$w_v^\curlyvee:=w^V_v -\|v\|^{-2}\langle v,w\rangle_q W_v  \text{ and }  w_v^\curlywedge:= w^H_v -\|v\|^{-2} dU(w) W_v. $$ 
If $\{w_1, \dots, w_m\}$ is  a basis of $T_qM$ and $v\in N(\mathcal{E})$ for some $\mathcal{E}\in \mathbbm{R}$,
then $\{w^\curlywedge_1, \dots, w^\curlywedge_m, w^\curlyvee_1, \dots, w^\curlyvee_{m-1}\}$
is a basis for $T_vN(\mathcal{E})$, providing decompositions  
$$TN(\mathcal{E})=T^\curlyvee N\oplus T^\curlywedge N  \text{ and } 
TN=T^\curlyvee N\oplus T^\curlywedge N \oplus \mathbbm{R}\eta,$$ where $T^\curlyvee N$ and $T^\curlywedge N$ 
  are spanned by vectors of the form $w^\curlyvee$ and $w^\curlywedge$, respectively.
Define the  forms $\omega_{ij}(w):=\langle \nabla_w e_j, e_i\rangle$.
Letting $(q,v)=F(q,u,\mathcal{E})$,  then
$$dF_{(q,u,\mathcal{E})} c_i^\curlyvee= h_{\mathcal{E}}(q) e_i^\curlyvee(v) \text{ and } dF_{(q,u,\mathcal{E})} e_i = e_i^\curlywedge(v) + \sum_{r,s=1}^m\omega_{sr}(e_i)u_r e_s^\curlyvee(v). $$
The vector field $\eta$ transforms under $F$ according to 
$$ dF^{-1}_v \eta=\left(\frac{\text{grad}_qU}{\|v\|^2+\|dU\|_q^2}, \left(\frac{\|v\|_q^{-1}\sum_j \langle v,e_j\rangle_q \omega_{ji}(\text{grad}_qU)}{\|v\|^2+\|dU\|_q^2}\right)_i, 1\right).$$  
\end{proposition}
\begin{proof}
We only  obtain $dF^{-1}_v w^V$ and $dF^{-1}_v w^H$ to illustrate the method of calculation. First, $dF^{-1}_v w^V$ equals
$$\left.\frac{d}{ds}\right|_{s=0} F^{-1}(v+sw)=\left.\frac{d}{ds}\right|_{s=0} \left(q, \left(\frac{ \langle v+sw, e_i\rangle_q}{\|v+sw\|_q}\right),
\frac12\|v+sw\|^2_q+U(q)\right)=(0, (\xi_i), \langle v, w\rangle_q),$$
where 
$\xi_i=\|v\|^{-1}\langle w-\langle v,w\rangle v/\|v\|^2,
e_i\rangle=\|v\|^{-1}\langle K_vw^\curlyvee, e_i\rangle$. Before calculating $dF^{-1}_v w^H$, first note
that $$w_v^H=\left. \frac{d}{ds}\right|_{s=0}\mathcal{P}_{\gamma(s)}v,$$
where $\gamma(s)$ is a differentiable curve such that $\gamma(0)=q$ and $\gamma'(0)=w$, and $\mathcal{P}_{\gamma(s)}v$
indicates the parallel translation of $v$ along $\gamma$. Keeping in mind that $\|\mathcal{P}_{\gamma(s)}v\|=\|v\|$  
and that  $ \left.\frac{d}{ds}\right|_{s=0}\langle \mathcal{P}_{\gamma(s)}v,e_i\rangle=\langle v, \nabla_w e_i\rangle$,
we obtain
\begin{align*} \left.\frac{d}{ds}\right|_{s=0}F^{-1}(\mathcal{P}_{\gamma(s)}v)&=\left.\frac{d}{ds}\right|_{s=0}\left(\gamma(s),\left(\frac{\langle\mathcal{P}_{\gamma(s)}v,e_i\rangle}{\|\mathcal{P}_{\gamma{s}}v\|}\right),  \frac12\|\mathcal{P}_{\gamma(s)}v\|^2 + U(\gamma(s))\right)\\
&=\left(w, \left(\|v\|^{-1}\sum_j \langle v, e_j\rangle \omega_{ji}(w)\right), dU(w) \right).
\end{align*}
The claimed identities are easily obtained from these. 
\end{proof}

  We have so far made no special assumptions about   the local  orthonormal frame $\{e_1, \dots, e_m\}$. 
Since we may   want to consider  the invariant volume form  $\Omega^E$    near boundary points of $M$,
 it makes sense to introduce the following concept: The orthonormal frame is  said to be {\em adapted}
 to a  codimension-$1$  foliation $\mathcal{S}$ of $U\subset M$ if $\{e_1, \dots, e_{m-1}\}$  spans the
 tangent space to each leaf $S$ of $\mathcal{S}$ at any given point $q\in U$.  
 Recall that the set of elements in $N$ (respectively, in $N(\mathcal{E})$) with base point  in $ S$ is  denoted by $N_S$ (respectively, $N_S(\mathcal{E})$).  The set
  $\{e_1^\curlywedge, \dots, e_{m-1}^\curlywedge, e_1^\curlyvee, \dots, e_{m-1}^\curlyvee\}$ 
 is easily seen to
 be a frame on $N_S(\mathcal{E})$.  It was noted before that
$\{e_1^\curlywedge, \dots, e_m^\curlywedge, e_1^\curlyvee, \dots, e_{m-1}^\curlyvee\}$
 is a local, not necessarily orthonormal,   frame on $N(\mathcal{E})$, and it can be shown  exactly as in Proposition \ref{tangent}  that 
 for a tangent vector to $N(\mathcal{E})$ to actually be tangent to $N_S(\mathcal{E})$ it is necessary and sufficient that
 its projection be tangent to $S$.  In particular, if $\{e_1, \dots, e_m\}$ is an adapted frame, the distribution in $TN(\mathcal{E})$ spanned
  by   $\{e_1^\curlywedge, \dots, e_{m-1}^\curlywedge, e_1^\curlyvee, \dots, e_{m-1}^\curlyvee\}$  is involutive.

\begin{proposition}\label{propositionTheo4}
Define  on $N_U\setminus \{\text{zero section}\}$ the functions 
  $$\psi_i(v):=\langle v, e_i\rangle_q\text{ and }   \psi_0(v):=\left(\|v\|^2 + \|\text{grad }U\|^2\right)^{-1} dU_q(v),$$
  where $\{e_1, \dots, e_m\}$ is an orthonormal frame on $U$.
    The Hamiltonian vector field $X^E$ has the form 
 \begin{equation}\label{prop1}X^E=\sum_{i=1}^m  \psi_i e_i^\curlywedge - \sum_{i=1}^mdU(e_i) e_i^\curlyvee=\sum_{i=1}^m \psi_i e_i^\curlywedge -\sum_{i=1}^{m-1}\left(dU(e_i) - \frac{dU(e_m)}{\psi_m}\psi_i\right) e_i^\curlyvee. \end{equation}
  The contact form $\theta$ restricted to $N_U$  can be written as
  \begin{equation}\label{prop2}\theta_v =\psi_0(v) \eta^* +\sum_{i=1}^m \psi_i(v) \epsilon_i^\curlywedge,\end{equation}
where   $\{\eta^*, \epsilon_1^\curlywedge, \dots, \epsilon_m^\curlywedge, \epsilon_1^\curlyvee, \dots, \epsilon_{m-1}^{\curlyvee}\}$
  is  the dual basis of   $\{\eta, e_1^\curlywedge, \dots, e_m^\curlywedge, e_1^\curlyvee, \dots, e_{m-1}^{\curlyvee}\}$.  
  The volume form $\Omega^E$ on each $N(\mathcal{E})$  over the set  $U\cap M^\mathcal{E}$  can be written as
  \begin{equation}\label{prop3}\Omega^E={\psi^{-1}_m} \epsilon_m^\curlywedge \wedge \left(X^E \righthalfcup \Omega^E\right)= {m}{{\psi^{-1}_m}} \epsilon_m^\curlywedge \wedge(d\theta)^{m-1}. \end{equation}
  Now suppose that  $\{e_1, \dots, e_m \}$ is adapted to a local codimension-$1$ foliation $\mathcal{S}$ and let 
 \begin{equation}\label{prop4}\iota_{\mathcal{S}}^* \theta=\psi_1 \epsilon_1^\curlywedge +\dots + \psi_{m-1}\epsilon_{m-1}^{\curlywedge}\end{equation}
  be the restriction of $\theta$ to the leaves $N_S(\mathcal{E})$. Then $\iota_{\mathcal{S}}^* d\theta$ is a symplectic form on each $N_S(\mathcal{E})$,
  and it  can be written as
 \begin{equation}\label{prop5}\iota_{\mathcal{S}}^* d\theta=-\sum_{i, j=1}^{m-1}\left(\delta_{ij} -\frac{\psi_i \psi_j}{h^2_{\mathcal{E}}}\right) \epsilon_i^\curlywedge\wedge \epsilon_j^\curlyvee+   \sum_{i, j=1}^{m-1} \frac{
  dU(e_j)\psi_i - dU(e_i)\psi_j}{h_{\mathcal{E}}^{2}}      \epsilon_i^\curlywedge\wedge \epsilon_j^\curlywedge  \end{equation}
  It follows that 
  \begin{equation}\label{prop6}
  \iota^*_{\mathcal{S}}(d\theta)^{m-1}=\pm \frac{\psi_m^2}{h_{\mathcal{E}}^2}\epsilon_1^\curlywedge\wedge \cdots \wedge \epsilon_{m-1}^\curlywedge \wedge \epsilon_1^\curlyvee \wedge \cdots \wedge \epsilon_{m-1}^\curlyvee , \ \Omega^E = \pm m\frac{\psi_m}{h^2_{\mathcal{E}}}\epsilon_1^\curlywedge\wedge \cdots \wedge \epsilon_{m}^\curlywedge \wedge \epsilon_1^\curlyvee \wedge \cdots \wedge \epsilon_{m-1}^\curlyvee.
  \end{equation}
  The volume form $\Omega^E$ transforms under the diffeomorphism $F:\mathcal{N}_U(\mathcal{E})\rightarrow N_U(\mathcal{E})$
  defined by \ref{diffeomorphism}
  according to 
  $$ F^*\Omega^E=\pm m h_{\mathcal{E}}^{m-2}\pi_1^* \omega^M \wedge \pi_2^*\omega^{\text{\tiny sphere}}$$
  The symplectic form on a hypersurface $S$ in $U$  is expressed under $F$ according to
  $$F^*(d\theta)^{m-1}= \pm  \left(\frac{\psi_m\circ F}{h_\mathcal{E}}\right) h_{\mathcal{E}}^{m-1}\pi_1^* \omega^S \wedge \pi_2^*\omega^{\text{\tiny sphere}}$$
\end{proposition}
\begin{proof}
All of this follows straightforwardly from the definitions and basic facts.
We only make a few comments. Identity \ref{prop3} results by noting that
$\epsilon_m^\curlywedge \wedge \left(X^E \righthalfcup \Omega^E\right)$ is a $(2m-1)$-form on $N(\mathcal{E})$,
thus it can be written as $f\Omega^E$, where  the function $f$ is found by applying the interior multiplication with $X^E$ and
using   that  $\psi_m$ is the coefficient of  $X^E$ for the basis element $e_m^\curlywedge$. Item 2 of 
Proposition \ref{XintoOmega}  is  also needed. Identity \ref{prop5} can be derived with little effort by using 
the identity \ref{sympmetric}, which expresses  the symplectic form $d\theta$ in terms of the 
Sasaki metric. Identity \ref{prop6} is a consequence of \ref{prop5} and the identity
$\det(I + ab^t)=1 + b^ta$, where $I$ is the identity matrix,  $a, b$ are column vectors, and $b^t$ is the   
row vector associated to $b$ after transpose.  It should be kept in mind that the $e_i^\curlyvee$ span an $(m-1)$-dimensional subspace at each point, so they are linearly dependent. In
  fact, they satisfy the equation $\sum_{i=1}^m \psi_i e_i^\curlyvee=0$.
\end{proof}

  Theorem \ref{invariantvol} is a corollary of the proposition.

\end{document}